\documentclass[conference]{IEEEtran}
\usepackage{graphicx}

\usepackage[]{caption}
\DeclareCaptionFont{mysize}{\fontsize{8}{9.6}\selectfont}
\usepackage[justification=centering,font=mysize]{subcaption}

\usepackage[justification=centering]{subcaption}
\usepackage{cite}

\usepackage{amsmath}
\usepackage{amssymb}

\usepackage{amsthm}
\usepackage[ruled,linesnumbered]{algorithm2e}
\usepackage{bbm}
\usepackage{makecell}
\usepackage{multirow}
\usepackage{pbox}
\usepackage{url}
\usepackage{xspace}
\usepackage[htt]{hyphenat}

\usepackage{tikz}
\usetikzlibrary{positioning}
\usetikzlibrary{calc}
\usepackage{pgfplots}
\usepgflibrary{shapes.geometric}
\usetikzlibrary{patterns,arrows}

\usepackage{alltt}
\usepackage[normalem]{ulem}

\let\labelindent\relax
\usepackage{enumitem}
\usepackage{textcomp}

\theoremstyle{definition}
\newtheorem{mydef}{Definition}

\newtheorem{theorem}{Theorem}
\newtheorem{lemma}{Lemma}

\DeclareMathOperator{\rspn}{rsp}
\DeclareMathOperator{\pointloss}{point-loss}
\DeclareMathOperator{\loss}{Loss}
\newcommand{\etal}{{et al.}}

\newcommand{\eg}{{e.g.}}
\newcommand{\ie}{{i.e.}}
\newcommand{\ignore}[1]{}
\newcommand{\ph}[1]{\noindent \textbf{#1} ---}

\definecolor{orange}{rgb}{0.9,0.7,0.3}
\newcommand{\tofix}[1]{{#1}}
\newcommand{\barzan}[1]{}
\newcommand{\young}[1]{}
\newcommand{\mike}[1]{}

\newcommand{\delete}[1]{}
\newcommand{\addnew}[1]{#1}

\newcommand{\greedy}[0]{{\tt\small Greedy}\xspace}
\newcommand{\interchange}[0]{{\sf\small Interchange}\xspace}

\newcommand{\vas}[0]{{\sc VAS}\xspace}

\IEEEoverridecommandlockouts
\begin{document}

\pgfplotsset {
    every axis/.append style={font=\scriptsize}
}

\pgfplotsset{
compat=newest,
every axis legend/.append style={font=\scriptsize, column sep=5pt},
/pgfplots/ybar legend/.style={
    /pgfplots/legend image code/.code={%
        \draw[##1,/tikz/.cd, bar width=6pt, yshift=-0.25em, bar shift=0pt, xshift=0.8em]
    plot coordinates {(0cm,0.8em)};}
}
}

%\title{Bare Demo of IEEEtran.cls for Conferences}
\title{Visualization-Aware Sampling\\ for Very Large Databases\thanks{\noindent This work
has been published in ICDE (IEEE International Conference on Data Engineering)
2016. 978-1-5090-2020-1/16/\$31.00 \textcopyright\ 2016 IEEE} \vspace{-6mm}}

\author{
\IEEEauthorblockN{Yongjoo Park, Michael Cafarella, Barzan Mozafari}
\IEEEauthorblockA{University of Michigan, Ann Arbor, USA\\
\{pyongjoo, michjc, mozafari\}@umich.edu}
%\and
%\IEEEauthorblockN{}
%\IEEEauthorblockA{Computer Science and Engineering\\
%University of Michigan\\
%Ann Arbor, USA\\
%Email: }
%\and
%\IEEEauthorblockN{}
%\IEEEauthorblockA{Computer Science and Engineering\\
%University of Michigan\\
%Ann Arbor, USA\\
%Email: }
\vspace{2mm}
}


%Telephone: (800) 555--1212, Fax: (888) 555--1212}


\maketitle

\begin{abstract}
Interactive visualizations are crucial in \emph{ad hoc} data  exploration and
analysis.  However, with the growing number of massive datasets, generating
visualizations in interactive timescales is increasingly challenging.  One
approach for improving the speed of the visualization tool is via {\em data
reduction} in order to  reduce the computational overhead, but at a potential
cost in visualization accuracy.  Common data reduction techniques, such as
uniform and stratified sampling, do not exploit the fact that the sampled tuples
will be transformed into a visualization for human consumption.

We propose a {\em visualization-aware} sampling (\vas) that guarantees high
quality visualizations with a small subset of the entire dataset.  We validate
our method when applied to scatter and map plots for three common visualization
goals: regression, density estimation, and clustering.  The key to our sampling
method's success is in choosing a set of tuples that minimizes a visualization-inspired
loss function. While existing sampling approaches minimize
the error of aggregation queries, we
 focus on a loss function that maximizes the visual fidelity of scatter plots.
Our user study confirms that our proposed loss function
correlates strongly with user success in using the resulting visualizations. 
Our experiments show that (i) \vas\ improves user's success by up to
$35\%$ in various visualization tasks, and (ii) \vas\ can achieve a required
visualization quality up to $400\times$ faster.

\end{abstract}


% \fi
\IEEEpeerreviewmaketitle

% New section
 %!TEX root = viz_icde16.tex

\section{Introduction}
\label{sec:intro}

%\barzan{The purpose of scientific visualization is to graphically illustrate
%scientific data to enable scientists to understand, illustrate, and glean
%insight from their
%data~\footnote{\url{http://en.wikipedia.org/wiki/Scientific_visualization}}.}

Data scientists frequently rely on visualizations for analyzing data
and gleaning insight. For productive data exploration, analysts should
be able to produce \emph{ad hoc} visualizations in interactive time (a
well-established goal in the visualization and human-computer interaction (HCI) community~\cite{cottam:2013,
heer:2012, wickham:2013, liu2013immens, lins:2013, barnett:2013,
fisher:2012, fisher:2012b, cottam:2010, piringer:2009}).
However, with the rise of big data and the growing number of databases with millions or even billions of records,
 generating even simple visualizations 
 %such as panning out or zooming in on visualizations 
 can take a considerable amount of time.  For example, as reported in Figure \ref{fig:intro_latency}, 
 we found that
 the industry standard Tableau visualization system takes over
 $4$ minutes on a high-end server to generate a scatterplot for a
 $50$M-tuple dataset that is already resident in memory.  (see Section \ref{sec:latency} for experimental details.)
  On the other hand, HCI
 researchers have found that visualizations must be generated in 500ms
 to 2 seconds in order for users to stay engaged and view the system as
 interactive \cite{miller1968response, shneiderman1984response,
 liueffects}.  Unfortunately,  dataset sizes are already growing faster than Moore's Law~\cite{ionurl} 
 (the rate at which our hardware is speculated to improve), 
 so technology trends will likely exacerbate rather than alleviate the problem.

This paper addresses the problem of interactive visualization in the
case of {\em scatterplots} and {\em map plots}.  Scatterplots are a
well-known visualization technique that represent database records
using dots in a 2D coordinate system.  For example, an engineer may
investigate the relationship between {\sf\small time-of-day} and
{\sf\small server-latency} by processing a database of Web server logs,
setting {\sf\small time-of-day} as the scatterplot's X-axis and the
{\sf\small server-latency} as its Y-axis.  Map plots display geographically-tied
values on a 2D plane.  Figure~\ref{fig:intro}(a) is an example of a
map plot, visualizing a GPS dataset from OpenStreetMap
project with
2B data points, each consisting of a  {\sf\small latitude}, {\sf\small
longitude},  and {\sf\small altitude} triplet ({\sf\small altitude} encoded with color).

\begin{figure}[t]
  \centering
  \begin{subfigure}[b]{0.48\columnwidth}
    \centering
    \includegraphics[width=1.1\columnwidth,height=31mm]{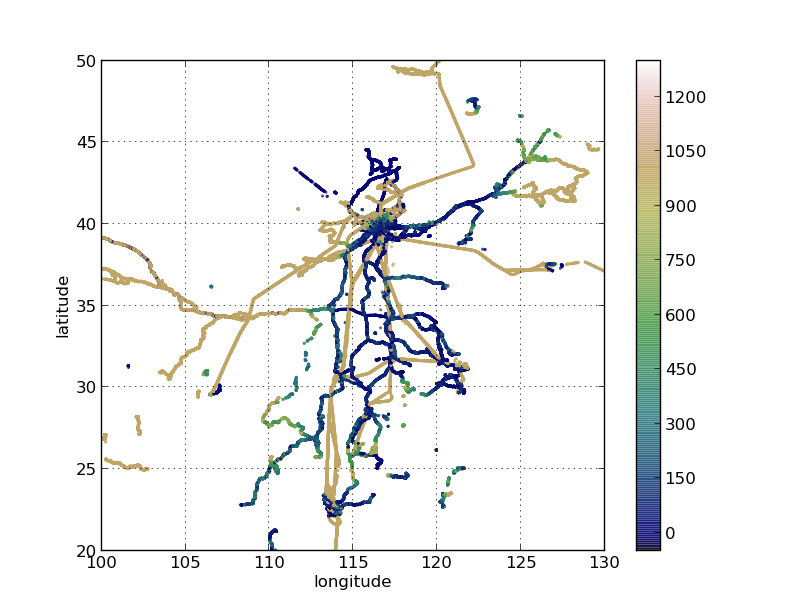}
    \caption{Stratified Sampling\\(overview)}
  \end{subfigure}
  ~
  \begin{subfigure}[b]{0.48\columnwidth}
    \includegraphics[width=1.1\columnwidth,height=31mm]{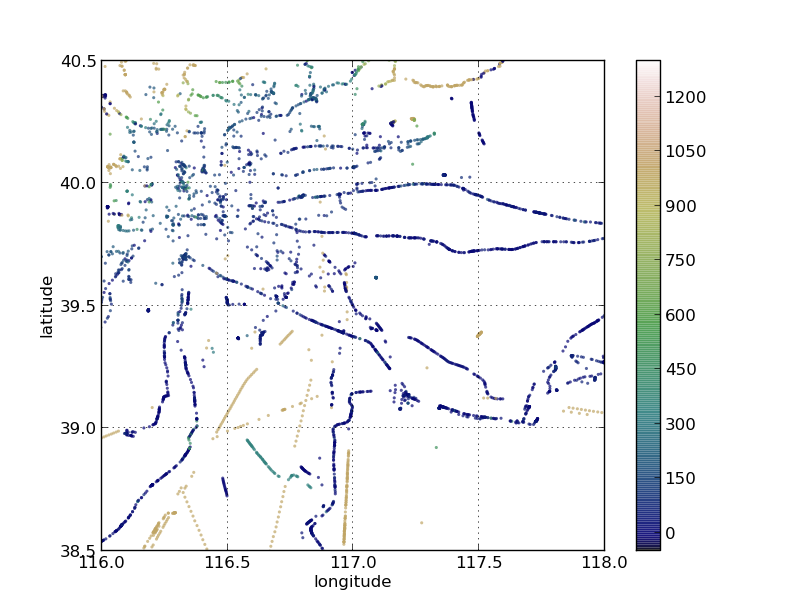}
    \caption{Stratified Sampling\\(zoom-in)}
  \end{subfigure}

  \begin{subfigure}[b]{0.48\columnwidth}
    \centering
    \includegraphics[width=1.1\columnwidth,height=31mm]{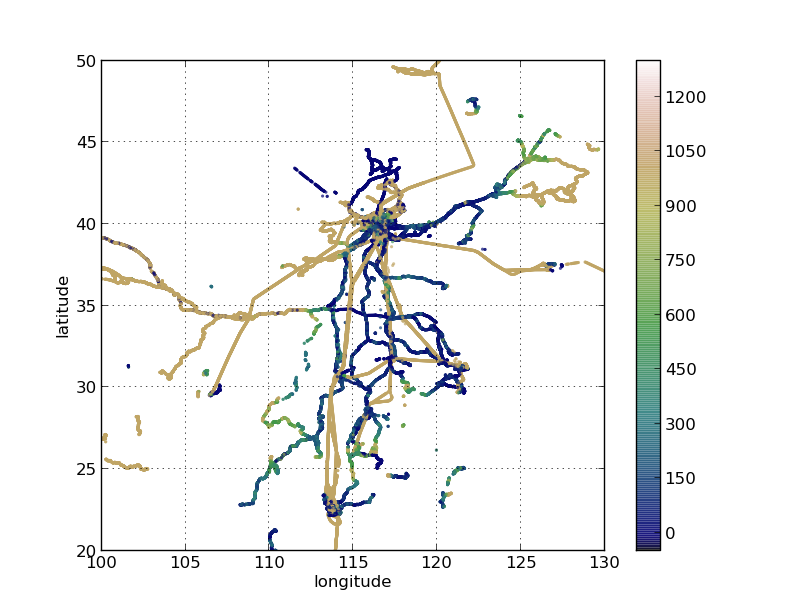}
    \caption{VAS (overview)}
  \end{subfigure}
  ~
  \begin{subfigure}[b]{0.48\columnwidth}
    \includegraphics[width=1.1\columnwidth,height=31mm]{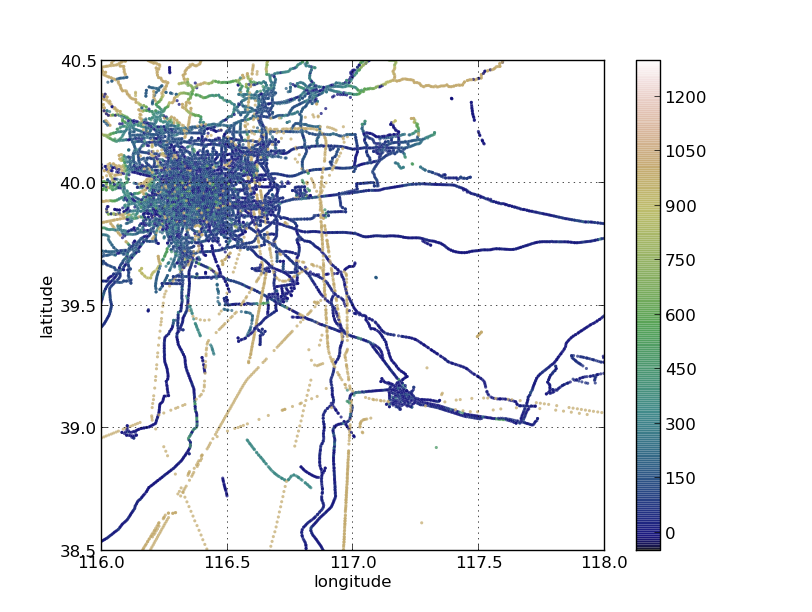}
    \caption{VAS (zoom-in)}
  \end{subfigure}

  \caption{Samples generated by fined-grained stratified sampling and our
    approach respectively. When the entire range is visualized, both methods seem to offer the
    visualization of the same quality. However, when zoomed-in views were requested,
    only our approach retained important structures of the database.}
  \label{fig:intro}
\end{figure}

%\begin{figure}[t]
%	\centering
%
%    \begin{minipage}[b]{0.5\textwidth}
%    \centering
%
%    \begin{subfigure}[t]{0.48\textwidth}
%        \centering
%        \includegraphics[width=\textwidth,height=0.8\textwidth,natwidth=510,natheight=376]{figs/open_all_india.png}
%        \caption{Original 2B-tuple dataset, visualized in 71 mins.}
%    \end{subfigure}
%    ~
%    \begin{subfigure}[t]{0.48\textwidth}
%        \centering
%        \includegraphics[width=\textwidth,height=0.8\textwidth,natwidth=510,natheight=375]{figs/open_rs_1M_india.png}
%        \caption{1M-tuple uniform sample, visualized in 3 secs.}
%    \end{subfigure}
%
%    \vspace{2mm}
%
%    \begin{subfigure}[t]{0.48\textwidth}
%        \centering
%        \includegraphics[width=\textwidth,height=0.8\textwidth,natwidth=510,natheight=376]{figs/open_ss_1M_india.png}
%        \caption{1M-tuple stratified sample, visualized in 3 secs.}
%    \end{subfigure}
%    ~
%    \begin{subfigure}[t]{0.48\textwidth}
%        \centering
%        \includegraphics[width=\textwidth,height=0.8\textwidth,natwidth=509,natheight=379]{figs/open_subs_1M_india.png}
%        \caption{1M-tuple \vas\ sample, visualized in
%        3 secs.}
%    \end{subfigure}
%
%    \end{minipage}
%
%    \caption{Different sampled map plots of a database of India GPS
%    device trace data from the OpenStreetMap project.  Each colored
%    dot represents the elevation sensed at a specific latitude and
%    longitude.} 
%    \label{fig:intro}
%\end{figure}

One approach for reducing the time involved in visualization production
is via {\em data reduction}~\cite{battle:2013}.  Reducing the dataset
 reduces the amount of work for the visualization system (by
reducing the CPU, I/O, and rendering times) but at
a potential cost to the quality of output visualization.
Effective data reduction
will shrink the dataset as much as possible while still producing an
output that preserves all important information of the original
dataset.
Sampling is a popular database method for reducing the amount of data
to be processed, often in the context of approximate query processing~\cite{mozafari_eurosys2013,mozafari_cidr2015,mozafari_sigmod2014_abm,acharya1999aqua,babcock2003dynamic,chaudhuri2007optimized,hellerstein1997online,jermaine2008scalable,olston2009interactive}.
%However, existing techniques have prominent limitations that deter widespread adoption (see Section \ref{sec:related}). 
While uniform (random) sampling and stratified sampling are two of the
most common and effective approaches in approximate  query processing~\cite{approx_chapter},
they are not well-suited for generating scatter and map plots: they can both
fail to capture important features of the data if they are sparsely
represented~\cite{liu2013immens}.

Figure~\ref{fig:intro} depicts an example using the Geolife
  dataset~\cite{zheng2008understanding}.
This dataset contains GPS-recorded locations visited by the people living
in and around Beijing. In this example, we visualized 100K datapoints using both
stratified sampling and our approach.  For stratified
sampling, we created a 316-by-316 grid and set the strata sizes (the number of
datapoints in each cell) as
balanced as possible across the cells created by the grid. 
In the zoomed-out overview plots, the visualization quality of the two
competing methods seem nearly identical; however, when a zoomed-in plot is generated,
one can observe  that our proposed method delivers significantly richer information.

%\begin{table}[!t]
%\centering
%\small
%
%\begin{tabular}{|l|r|l|r|}
%\hline
%\multicolumn{2}{|c|}{Tableau} & \multicolumn{2}{c|}{MathGL} \\ \hline
%Dataset Size & Latency & Dataset Size & Latency    \\ \hline
%1M           & 4s      & 1M          & 9s          \\
%5M           & 24s     & 5M          & 45s         \\
%10M          & 48s     & 10M         & 1m 29s      \\
%50M          & 5m 22s  & 50M         & 7m 30s      \\ \hline
%\end{tabular}
%
%\caption{The latency for scatter plot visualization. Tableau and MathGL were
%used for benchmark.}
%\label{tab:intro_latency}
%\end{table}

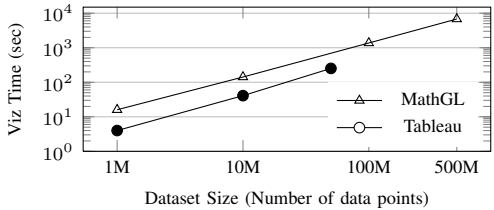
\begin{figure}[!t]
\centering

\begin{tikzpicture}
\begin{axis}[
width=70mm,
height=35mm,
xtick={1e6,1e7,1e8,5e8},
xticklabels={1M, 10M, 100M, 500M},
ytick={1,10,100,1000,1e4},
ymin=1,
ymax=15000,
xmode=log,
ymode=log,
xlabel=Dataset Size (Number of data points),
ylabel=Viz Time (sec),
legend style={draw=none,legend pos=south east},
ymajorgrids,
]
\addplot[mark=triangle]
table {
x	y
1000000	16
10000000	142
100000000	1373
500000000	6795
};
\addplot[mark=*]
table {
x	y
1000000	4
10000000	41
50000000	252
};
\addlegendentry{MathGL}
\addlegendentry{Tableau}
\end{axis}
\end{tikzpicture}
\caption{The latency for generating scatter plot visualizations using Tableau and MathGL (a library for scientific graphics).}
\label{fig:intro_latency}
\end{figure}

\ph{Previous Approaches}  Architecturally, our system is similar
to ScalaR~\cite{battle:2013}, which interposes a data reduction layer
between the visualization tool and a database backend; however, that
project uses simple uniform random
sampling. 
 Researchers have attempted a
number of approaches for improving visualization time, including
binned aggregation~\cite{lins:2013, liu2013immens, wickham:2013},
parallel rendering~\cite{cottam:2010, piringer:2009}, and incremental
visualization~\cite{fisher:2012, fisher:2012b}.  These methods are
orthogonal to the one we propose here.

\ph{Our Goals} This paper tackles the technical
challenge of creating a sampling strategy that will yield \emph{useful and
high-quality} scatter and map plots at \emph{arbitrary zooming resolutions}
with as \emph{few sampled tuples} as possible.  Figure \ref{fig:intro}(d)
shows the plot generated by our proposed method, which we call Visualization-Aware Sampling (\vas).  
Using the same number of tuples as random and stratified sampling, \vas\
yields a much higher-fidelity result. 
%\vl~provides an in-database data
%sampling operator, called
The use of \vas\ can be 
 specified as part of the queries
submitted by visualization tools to the database. Using \vas, the database returns an approximate
query answer within a specified time bound using
one of multiple pre-generated samples.
\vas\ chooses an appropriate
sample size by converting the specified
time bound into the number of tuples that can likely be processed
within that time bound.  \vas\ is successful because it samples data
points according to a visualization-oriented metric
that correlates well with user success across a range of
scatter and map plot tasks.

% I commented this out!
%In using such visualization tools, a considerable portion of a data scientist's workday is spent  on 
%    creating different scatterplots, zooming in and out, panning\footnote{Changing the focal point 
%    by sliding the scatter plot (or heat map) across its axes.}, slicing and dicing the data, or 
%    even changing the axes altogether. 

\ph{Contributions} We make the following contributions:
\begin{itemize}[nolistsep,itemsep=2pt]
\item We define the notion of \vas\  as an optimization problem
 (Section~\ref{sec:formulation}).
\item We prove that the \vas\ problem is NP-hard and an offer efficient
approximation algorithm. We establish a worst-case guarantee for our
approximate solution (Section~\ref{sec:solving}).  
\item In a user study, we show that our \vas\  is highly
correlated with the user's success rate in various visualization tasks.
We also evaluate the efficiency and effectiveness of our approximation
algorithm over several datasets.  We show
that \vas\ can deliver a visualization that has equal quality
with competing approaches, but using up to $400\times$  fewer data
points.  Alternatively, if \vas\ can process an equal number of data points
as competing methods, it can deliver a visualization with a significantly higher
quality (Section~\ref{sec:exp}).  
\end{itemize}

Finally, we cover related work in Section~\ref{sec:relatedwork} and
conclude with a discussion of future work in Section~\ref{sec:con}.

\ignore{
\vspace{0.4cm}
\noindent {\bf Our Approach ---}  
Through real-life user studies, we show that users are more
effective on a range of tasks when visualizing data samples that
were optimized for this metric.  We show that optimizing this metric
when sampling is NP-hard, and then propose an approximation algorithm.
We also propose a number of small optimizations that dramatically
speed up our sampling algorithm at a small cost in accuracy.  We have
implemented the method in a real middleware system used by a front-end
visualization package.

\vspace{0.4cm}
\noindent {\bf Contributions and Outline ---}  The central
contributions of this work include:
\begin{itemize}
  \item We propose a concrete user model of how scatterplot visualizations are
        used in practice.  (Section~\ref{sec:related}.)
  \item We describe an optimization problem that yields high-quality
        samples for our user model (Section~\ref{sec:formulation}).
        We further offer a an efficient approximate method for solving
        the optimization problem (Section~\ref{sec:solving}).
  \item We give an implmented sampling method that yields X improvement over
        random sampling, and Y improvement over stratified sampling.
        We also give user studies that demonstrate our system's
        visualizations are preferred over those from other methods.
        (Section~\ref{sec:exp}) 
\end{itemize}
}

 %!TEX root = viz_icde16.tex

\section{System Overview}
\label{sec:vizsystems}

\subsection{Software Architecture Model}

\begin{figure}[t]
  \begin{center}
    \includegraphics[height=1in]{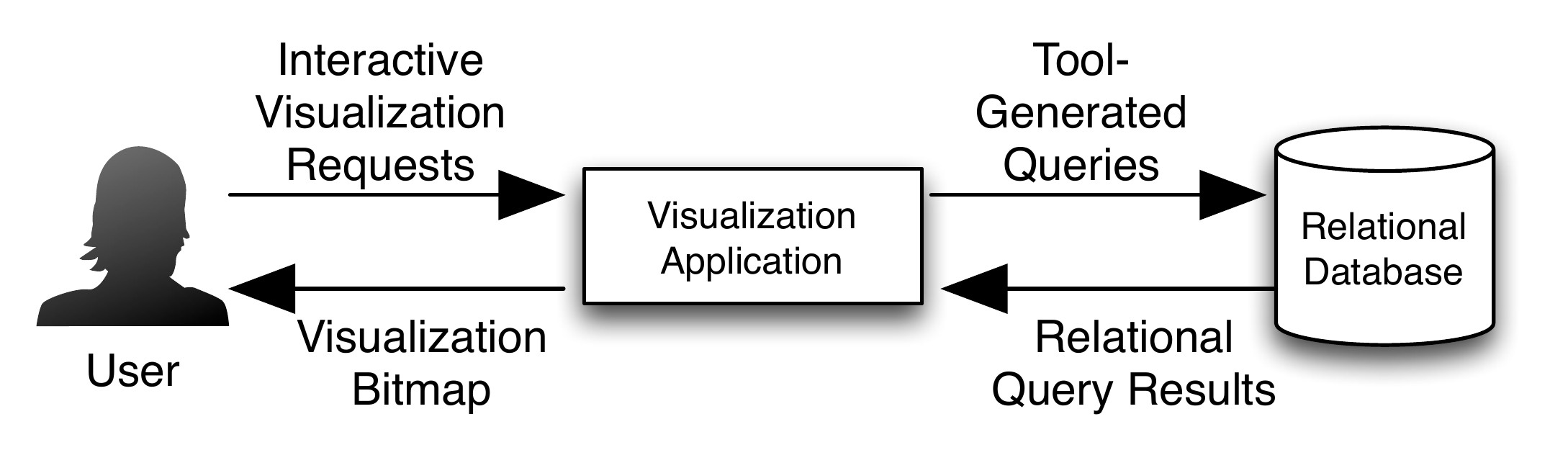}
  \end{center}    
    \caption{Standard model of user interaction with the combined
    visualization and database system.}
    \label{fig:architecture}
\end{figure}

Figure~\ref{fig:architecture} shows the  
software architecture model that we focus on in this paper. 
This is a standard architecture supported by the popular Tableau
system~\cite{tableauarchi}. 
 It is also similar to
ScalaR's ``dynamic reduction'' software architecture~\cite{battle:2013}.  The user interacts with a
visualization tool to describe a desired visualization --- say, a scatterplot of
{\sf\small Web server time} vs {\sf\small latency}.  This tool has been configured to
access a dedicated RDBMS, and the schema information from the RDBMS is
visible in the user's interaction with the tool.  For example, in
Tableau, a user can partially specify a desired scatterplot by
indicating a column name (say, {\sf\small server latency}) from a list of
options populated from the RDBMS metadata.  The user must not only choose
just which fields and ranges from the database are rendered, but also choose
image-specific parameters such as visualization type, axis labels,
color codings, and so on.   

Once the user has fully specified a visualization, the
tool requests the necessary data by generating the appropriate SQL query and submitting it to
 the remote RDBMS.  The RDBMS then returns the (relational) query
results back to the visualization tool.  Finally, the tool uses the fetched records
to render the final visualization bitmap, which is displayed to
the user.  During these last three steps, the user waits idly for the
visualization tool and the RDBMS.

When large datasets are being visualized, extremely long waits can negatively affect the analyst's level of engagement and ability to interactively produce successive visualizations
\cite{heer:2012, barnett:2013,miller1968response, shneiderman1984response, liueffects}. 
 As reported in
Figure \ref{fig:intro_latency}, our own experiments show that the
industry-standard Tableau tool can take more than four minutes to
produce a scatterplot on just 50M tuples fetched from an in-memory database.

Note that our sampling approach is not limited to the 
 software architecture in Figure~\ref{fig:architecture},  as 
 reducing the number of visualized records  almost always brings performance benefits.
Thus,  
%We understand that visualization  systems are still in their relative infancy; however,
 even if engineers decide to combine the visualization and
data management layers  in the future, sampling-based methods will still be
useful.

\subsection{Data Sampling}
\label{sec:datasampling}

Approximate query processing via sampling is a popular technique 
\cite{mozafari_cidr2015,mozafari_pvldb2012,acharya1999aqua,babcock2003dynamic,chaudhuri2007optimized,hellerstein1997online,jermaine2008scalable,olston2009interactive}
for reducing the number of returned records, and random sampling or stratified sampling are two well-known methods for this.
  When
using these methods, the visualization tool's query is run over a
{\em sampled table(s)} (or simply, a \emph{sample}) that is smaller than, and derived from, the
original table(s).  The sample(s) can be maintained by the same
RDBMS.  Since sampling approaches incur
an additional  overhead to produce the sample(s), 
these are typically performed in an \emph{offline} manner \cite{mozafari_eurosys2013,chaudhuri2007optimized}: 
Once the sample is created and stored in the database, they can be interactively queried and visualized many times. 
(A sample can also be periodically updated when  new data arrives \cite{mozafari_pvldb2012}.)

There is, of course, a tradeoff between output result quality and the
size of the sample\footnote{Here, the size of a sample means the
number of the data points contained in the sample.} (and thus, runtime).  In the limit, a random sample of 100\% of the original
database will produce results with perfect fidelity, but will also not
yield any reduction in runtime.  Conversely, a sample of 0\% of the
database will yield a result with no fidelity, albeit very quickly.
The exact size of the sample budget will be determined by
deployment-specific details: the nature of the application, the
patience of the user base, the amount of hardware resources available,
and so on.  As a result, the usefulness of a sampling method must be
evaluated over a range of sample budgets, with a realistic measure of
final output quality.  Choosing a correct sampling budget is a known
issue in the approximate query processing literature~\cite{mozafari_eurosys2013}.  Of course, in any specific real-world
deployment, we expect that the application will have a fixed maximum
runtime or the size of a sample that the system must observe.

\subsection{Visualization Quality}
\label{sec:viztasks}

In this work, we focus on the production of
{\em scatterplots} (including map plots, such as Figure~\ref{fig:intro}) as one of the most popular 
visualization techniques. We leave other  visualization types (such as bar charts,
line charts, chloropleths, and so on) to future work.

Since the final value of a visualization is how much it helps the user,
evaluating any sampling method means examining
how users actually employ the visualizations they produce.
Schneiderman, \etal\ proposed a taxonomy for information visualization
types~\cite{shneiderman1996eyes} and compiled some of  common
visualization-driven goals/tasks.  Their  list of goals included (i) regression,
(ii) density estimation, and (iii) clustering.  Our system aims to yield
visualizations that help with each of these  popular goals.
We make no claim about other user goals and tasks for
visualization, such as pattern finding and outlier detection, which we
reserve for future work (although we have anecdotal evidence to
suggest our system can address some of these tasks, too).
 
In this context, regression  is the task of (visually) estimating the value of dependent variables given the
value of independent variables. For example, if we want to know the temperature
of the location specified by a pair of latitude and longitude coordinates,
it belongs to the
regression task. Density estimation is the task of understanding the distribution
of the original data. For instance, one can use a map plot to understand the
geometric area with the most cell phone subscribers. Clustering is
a task that assigns data elements into distinct sets such that the data in
the same group tend to be close to one another, while data in
different groups are comparatively far apart.

Schneiderman, \etal's list also included goals/tasks that are either poor fits for scatter
 plots, or are simply outside the scope of what we aim to
accomplish in this paper: shape visualization (DNA or 3D structures), classification,
hierarchy understanding, and community detection in networks.  We
explicitly do {\em not} attempt to produce visualizations that can
help users with these tasks.

\subsection{Our Approach}
Our proposed method proceeds in two steps: (1) during {\em offline
preprocessing}, we produce a sample that enable fast
queries later, and (2) at {\em query time}, we choose a sample whose size is appropriate for the specific query.

Similar to any offline indexing technique, \vas also requires (1) the user to make
  choices about
indexed columns and (2) an upfront computational work to speed up future queries. 
In other words,  \vas can be considered as a specialized index designed for visualization workloads (\eg,
Tableau). Note that, currently, even if users want to use offline indexing, there is no
indexing technique that ensures fast and accurate visualizations, a problem
solved by \vas.

Indexed columns can be chosen in three ways:
\begin{enumerate}[nolistsep,noitemsep=2pt]
  \item manually, by the DBA;
  \item based on the most frequently visualized columns~\cite{mozafari_eurosys2013,idreos2007database}; or
  \item based on statistical properties of the data~\cite{parameswaran2013seedb}.
\end{enumerate}
Among these approaches, the second one is the simplest, works reasonably well in practice, 
and can be made resilient against workload
changes~\cite{mozafari_sigmod2015}. 
Furthermore, note that visualization workloads, especially
those supported by BI tools and SQL engines, are similar to exploratory SQL
analytics (i.e., grouping, filtering, aggregations).
Real-world traces from
Facebook and Conviva~\cite{mozafari_eurosys2013} reveal that 80-90\% of
exploratory queries use 5-10\% of the column combinations. Moreover, 
\vas only requires
frequently visualized column pairs, not groups or filters.

The core innovation of our work is that we generate a sample according to
a {\em visualization-specific metric}.  That is, we believe that when
a sample is generated according to the metric we propose in
Section~\ref{sec:formulation} below, the user will be able to
accomplish their goals from Section~\ref{sec:viztasks} above (i.e., regression, density estimation, clustering)
using the resulting visualization, 
even  with a small number of rendered data points.
  We do {\em not} claim that our
method will work for other visualization types, or even for other visualization goals.  
Indeed, our method of generating a sample could in principle be harmful
to some goals (such as community detection tasks that require
all members of a latent set to be sampled).
However, scatter plots, map plots, and the three visualization goals we focus on are quite widespread
and useful.  
Furthermore, modifying a visualization tool to only use our sampling method 
if a user declares an interest in one of these goals 
would be a straightforward task.

% New section
%\input{related}

% New section
%\input{systemoverview}

% New section
 %!TEX root = viz_icde16.tex

\section{Problem Formulation}
\label{sec:formulation}

The first step in our approach to obtain a good sample for scatter plot
visualizations is  defining a mathematical loss function that is closely
correlated with the loss of  visualization quality/utility from the user's
perspective.  Once we define a function that captures the visualization utility,
our goal will be to solve an optimization problem; that is, finding a  sample of
a given size that has the minimum value for the loss function.  In the rest of
this section, we formally define our loss function and optimization problem.
The algorithm for solving the derived optimization problem will be presented in
Section \ref{sec:solving}.  Also, in Section \ref{sec:user}, we  will report
a comprehensive user study confirming that minimizing our loss function does
indeed yield  more useful visualizations for various user tasks.

We start our problem formulation by defining notations.  We denote a dataset $D$
of $N$ tuples by $D = \{t_1, t_2, \ldots, t_N\}$.  Each tuple $t_i$
encodes the coordinate at which the associated \emph{point} is displayed.
For example, $t_i$ is a pair of longitude and latitude in a map plot.  A
sample $S$ is a subset of the dataset $D$ and is denoted by $S =  \{ s_1,
s_2, \ldots, s_K \}$. Naturally, $s_i$ is one of $t_j$ where $j = 1, \ldots, N$.
The size of the sample $S$ (which is denoted by $K$) is pre-determined based on the
interactive latency requirement (see Section~\ref{sec:datasampling}) and is
given as an input to our problem.

In designing our loss function, the objective is to measure the
visualization quality degradation originating from the sampling process.
The traditional goal of sampling in  database systems is to
maximize the number of tuples that match a selection predicate, particularly those
on categorical attributes~\cite{mozafari_eurosys2013}.
In contrast, the selection predicates
of a scatter/map plot are on a continuous range, for which
traditional approaches (e.g., uniform or stratified
sampling)
may not lead to high quality visualizations.

Therefore, to develop a more visualization-focused sampling technique, we first
imagine a 2D space on which a scatter plot is displayed, and let $x$ denote any
of the points on the space. To measure the visualization quality loss, we make
the following observations:
\begin{enumerate}
\item The visualization quality loss occurs, as the sample $S$ does not include all tuples of $D$.
\item The quality loss at $x$ is reduced as
    the sample includes points at or near $x$ --- two plots drawn using the
    original dataset ($D$) and the sample ($S$) might not look identical
    if $S$ does not include a point at $x$ whereas $D$ includes one at $x$,
    but they will look similar if the sample contains points near $x$.
\item When there are already many points at or near x, choosing more points in
  that neighborhood does not significantly enhance a visualization.
 \end{enumerate}

To express the above observations in a formal way, we consider the following
measure for visualization quality degradation at the point $x$:
\[
\pointloss(x) = \frac{1}{\sum_{s_i \in S} \kappa(x, s_i)}.
\]
where $\kappa(x, s_i)$ is the function that captures the \emph{proximity}
between the two points, $x$ and $s_i$.
In this paper, we use $\kappa(x, s_i) = \exp(-\|x - s_i\|^2 / \epsilon^2)$ (see
footnote\footnote{In our experiments, we set $\epsilon \approx
\max(\|x_i-x_j\|)/100$ but there is a theory on how to choose the optimal value
for $\epsilon$  as the only unknown parameter \cite{cizek2005statistical}.} for
$\epsilon$) although other functions can also be used for $\kappa$ if the function
is a decreasing convex function of $\|x - s_i\|$ --- the convexity is needed due
to the third observation we described above. The equivalent quality metric
can also be obtained by considering the problem as the regression problem that aims to
approximate the original tuples in $D$ using the sample $S$. See our technical
report for an alternative derivation~\cite{vas2015arxiv}. Note that the above
loss value is reduced if there exists more sampled points near $x$ where the
proximity to $x$ is captured by $\kappa(x, s_i)$. In other words, the
visualization quality loss at $x$ is minimized if $S$ includes as many points as
possible at and around $x$.

Note that the above loss function is defined for a single point $x$ on the space
on which a scatter/map plot is visualized, whereas the space on which a scatter plot
is drawn has many of those points. As a result, our goal should be to obtain a
sample $S$ that minimizes the \emph{combined} loss of all possible points on the
space. Due to the reason, our objective is to find a sample $S$ that minimizes the
following expression:
\begin{equation}
    \loss(S) = \int \pointloss(x) \;dx = 
    \int \frac{1}{\sum_{s_i \in S} \kappa(x, s_i)} \;dx
    \label{eq:loss}
\end{equation}
Here, the integration is performed over the entire 2D space.

Now, we perform several mathematical tricks to obtain an effectively equivalent
but a more tractable problem, because the exact computation of the above integration
requires an computationally expensive method such as a Monte Carlo experiment
with a large number of points.
Using a second-order Taylor expansion, we can obtain a simpler form that enables
an efficient algorithm in the next section:
\begin{align*}
&\min \int \frac{1}{\sum_{s_i \in S} \kappa(x, s_i)} \; dx \\
&= \min \int 1 - (\sum \kappa(x,s_i) - 1) + (\sum \kappa(x,s_i) - 1)^2 \; dx \\
&= \min \int (\sum \kappa(x,s_i))^2 - 3 \sum \kappa(x,s_i) \; dx \\
&= \min \int \sum_{s_i, s_j \in S} \kappa(x,s_i) \kappa(x,s_j) \; dx
\end{align*}

To obtain the last expression, we used the fact that the term $\int \sum
\kappa(x,s_i)\;dx$ is constant since $\kappa(x,s_i)$ is a similarity function
and we are integrating over every possible $x$, i.e., $\int \sum \kappa(x,s_i)\;dx$
has the same value regardless of the value of $s_i$.
For the same reason, $\int \sum [\kappa(x,s_i)]^2\;dx$ is also constant. By
changing the order of integration and summation, we obtain the following
optimization formulation, which we refer to as Visualization-Aware Sampling
(\vas) problem in this paper.

\begin{mydef}[\vas]
Given a fixed $K$, \vas~is the problem of obtaining a sample $S$ of size $K$
 as a solution to the following optimization problem:
\begin{align*}
&\min_{S \subseteq D;\; |S|=K} \sum_{s_i, s_j \in S;\; i < j}
\tilde{\kappa}(s_i,s_j) \\
&\text{where } \tilde{\kappa}(s_i, s_j) = \int \kappa(x,s_i) \kappa(x,s_j) dx
\end{align*}
\label{prob:vas}
\end{mydef}

In the definition above, we call the summation term $\sum
\tilde{\kappa}(s_i,s_j)$ the \emph{optimization objective}.
With our choice of the proximity function, $\kappa(s_i, s_j) =
\exp(-\|s_i-s_j\|^2/\epsilon^2)$,
we can obtain a concrete expression for $\tilde{\kappa}(s_i, s_j)$:
$\exp(-\|s_i - s_j\|^2/2\epsilon^2)$, after eliminating constant terms that do not
affect the minimization problem. In other words, $\tilde{\kappa}(s_i,s_j)$ is
in the same form as the original proximity function.
In general, $\tilde{\kappa}(s_i,s_j)$ is another proximity function between the
two points $s_i$ and $s_j$ since the integration for
$\tilde{\kappa}(s_i,s_j)$ tends to have a larger value when the two points are
close. Thus, in practice, it is sufficient to use any proximity function
directly in place of $\tilde{\kappa}(s_i,s_j)$. 
%For convenience, we will call $\tilde{\kappa}(s_i,s_j)$ itself a kernel function as well.

\delete{
Finally, note that our \vas\ problem has to be solved separately for every
independent column variable of interest.
In other words, similarly to stratified sampling
\cite{chaudhuri2007optimized,agarwal2013blinkdb}, users must declare the set of
columns that they want to plot \emph{a priori}, so that the \vas\
can be solved for each such column in an \emph{offline} manner.
While this may seem to limit the system's usability,  it still covers a wide
variety of real-world data analytics scenarios because certain columns tend to be
used much more frequently.  For example, \emph{time} tends to a popular choice 
in time series data while a pair of \emph{latitude} and \emph{longitude} will
likely be a popular choice in a map plot. In fact, Agarwal,
\etal~\cite{agarwal2013blinkdb} report that 90\% of the queries in Conviva and
Facebook can be answered using just 10\% or 20\% of the columns, respectively.}

In the next section, we show that the \vas\ problem defined above is NP-hard and we 
present an efficient approximation algorithm for solving this problem. 
Later in Section \ref{sec:user}
we show that by finding a sample $S$ that minimizes our loss function, we obtain a
sample that, when visualized, best allows users to perform various visualization tasks.

% New section

 %!TEX root = diversity_sampling_sigmod15.tex

\section{Solving VAS}
\label{sec:solving}

In this section, we focus on solving  the optimization problem derived in the previous
section to obtain an optimal sample $S$. In the following section, we also
describe how to extend the sample obtained by solving VAS to provide a richer set
of information.

\subsection{Hardness of VAS}

First, we analyze the hardness of \vas formally.

\begin{theorem}
\vas\ (Problem~\ref{prob:vas}) is NP-hard.
\end{theorem}

\begin{proof}
We show the NP-hardness of Problem~\ref{prob:vas} by reducing
\emph{maximum edge subgraph} problem to \vas.

\begin{lemma}
(\emph{Maximum Edge Subgraph}) Given a undirected weighted graph $G = (V, E)$,
choose a subgraph $G' = (V', E')$ with $|V'| = K$ that maximizes
\[
    \sum_{(u,v) \subset E'} w(u,v)
\]
This problem is called \emph{maximum edge subgraph}, and is
NP-hard~\cite{feige2001dense}.
\end{lemma}

To reduce the above problem to \vas, the following procedure is performed: map
$i$-th vertex $v_i$ to $i$-th instance $x_i$, and set the value of
$\tilde{\kappa}(x_i,x_j)$ to $w_{max} - w(v_i, v_j)$, where $w_{max} =
\max_{v_i, v_j \subset V'} w(v_i, v_j)$. The reduction process takes $O(|E| +
|V|)$. Once the set of data points that minimize $\sum_{s_i,s_j \in X}
\tilde{\kappa}(s_i,s_j)$ is obtained by solving \vas, we choose a set of
corresponding vertices, and return them as an answer to the \emph{maximum edge
subgraph} problem.  Since the \emph{maximum edge subgraph} problem is NP-hard,
and the reduction process takes a polynomial time, \vas~is also NP-hard.
\end{proof}

Due to the NP-hardness of \vas, obtaining an exact solution to \vas\ is
prohibitively slow,  as we will empirically show in Section~\ref{sec:exact}.
Thus, in the rest of this section, we present an approximation algorithm for
\vas\ (Section~\ref{sec:interchange}), followed by additional ideas for improvement
(Section~\ref{sec:speedup}).

%ta inja
\subsection{The Interchange Algorithm}
\label{sec:interchange}
\label{sec:speedup}

%Here we present an algorithm that returns a solution after a single-scan over
%the dataset. Due to the nature, this algorithm is faster than greedy algorithm
%in the above section. 

In this section, we present our approximation algorithm, called \interchange.
The \interchange algorithm starts from a randomly chosen set of size $K$ and
performs a replacement operation with a new data point if the operation
decreases the optimization objective (i.e., the loss function).  We call such a replacement, i.e., one
that decreases the optimization objective, a \emph{valid replacement}. In other
words,
\interchange tests for valid replacements as it sequentially reads through the
data points from the dataset $D$.

One way to understand this algorithm theoretically is by imagining a
Markov network in which each state represents a different subset of $D$ where
the size of the subset is $K$. The network then has a total of $\binom{D}{K}$
states. The transition between the states is defined as an exchange of one of
the elements in the current subset $S$ with another element in $D-S$. It is
easy to see that the transition defined in this way is \emph{irreducible},
i.e., any state can reach any other states following the transitions defined in
such a way. Because \interchange is a process that continuously
seeks a state with a lower optimization objective than the current one,
\interchange is a hill climbing  algorithm in the network.

\begin{algorithm}[t]
  \small
    \DontPrintSemicolon
    \SetKwInOut{Input}{input}
    \SetKwInOut{Output}{output}
    \SetKwFunction{Dist}{dist}
    \SetKwFunction{Expand}{Expand}
    \SetKwFunction{Shrink}{Shrink}
    \SetKwData{rsp}{rsp}
    \SetKwBlock{Proc}{}{end}
    \SetKw{Subroutine}{subroutine}

    \Input{$D = \{t_1, t_2, \ldots, t_N\}$}
    \Output{A sample $S$ of size $K$}
    \BlankLine
    \tcp{set for pairs of (item, responsibility)}
    $R$ $\leftarrow$ $\varnothing$ \;
    \ForEach{$t_i \in D$}{
        \lIf{$|R| < K$}{ $R \leftarrow$ \Expand($R, t_i$) }
        \Else{
            $R \leftarrow$ \Expand($R, t_i$) \;
            $R \leftarrow$ \Shrink($R$) \;
        }
    }
    $S \leftarrow$ pick the first item of every pair in $R$\;
    \Return $S$\;

    \BlankLine
    \Subroutine \Expand($R, t$)
    \Proc{
        \rsp $\leftarrow 0$ \qquad \tcp{responsibility}
        \ForEach{$(s_i, r_i) \in R$}{
            $l$ $\leftarrow \tilde{\kappa}(t, s_i)$\;
            $r_i$ $\leftarrow$ $r_i + l$\;
            \rsp $\leftarrow \rsp + l$\;
        }
        insert $(t, \rsp)$ into $R$\;
        \Return $R$\;
    }

    \BlankLine
    \Subroutine \Shrink($R$)
    \Proc{
        remove $(t, r)$ with largest $r$ from $R$\;
        \ForEach{$(s_i, r_i) \in R$}{
            $r_i$ $\leftarrow$ $r_i - \tilde{\kappa}(t, s_i)$\;
        }
        \Return $R$\;
    }

    \caption{Interchange algorithm.}
    \label{algo:interchange}
\end{algorithm}

%Algorithm~\ref{algo:interchange} shows the interchange algorithm to solve
%Problem~\ref{pb:max}. Different from the greedy algorithm in
%Section~\ref{sec:greedy}, we do not enlarge the set $S$ one by one. But instead
%we start with a set of size $K$ (by reading the first $K$ data from the
%dataset), and as reading the remaining data sequentially, we replace an old
%data with a new one if the operation decrease the total loss.

\ph{Expand/Shrink procedure}
Now we state how we can efficiently perform valid replacements.  One
approach to finding valid replacements is
by substituting one of the elements in $S$ with a new data point whenever 
one is
read in, then computing the optimization objective of the set.
For this computation, we need to
call the proximity function $O(K^2)$ times as there are $K$ elements in the
set, and we need to compute a proximity function for every pair of elements in
the set. This computation should be done for every element in $S$. Thus, to test
for valid replacements, we need $O(K^3)$ computations for every new data point.

A more efficient approach is to consider only the part of the optimization
objective for which the participating elements are \emph{responsible}. We
formally define the notion of responsibility as follows.
\begin{mydef}
(\emph{Responsibility}) The responsibility of an element $s_i$ in set $S$ is
defined as:
\[
\text{rsp}_S (s_i) = \frac{1}{2} \sum_{s_j \in S,\,j \ne i}
\tilde{\kappa}(s_i,s_j).
\]
\end{mydef}

Using the responsibility, we can speed up the tests
for valid replacements in the following way. Whenever considering a new data
point $t$, take an existing element $s_i$ in $S$, and compute the responsibility
of $t$ in the set $S - \{s_i\} + \{t\}$. This computation takes $O(K)$ times. It
is easy to see that if the responsibility of $t$ in $S - \{s_i\} + \{t\}$
is smaller than the responsibility of $s_i$ in the original set
$S$, the replacement operation of $s_i$ with the new data point $t$ is a \emph{valid
replacement}.  In this way, we can compare the responsibilities without
computing all pairwise proximity functions. Since this test should be performed for
every element in $S$, it takes a total of $O(K^2)$ computations for every new
data point.

However, it is possible to make this operation even faster. Instead of testing
for valid replacements by substituting the new data point $t$ for one of the elements in
$S$, we simply \emph{expand} the set $S$ by inserting $t$ into the set,
temporarily creating a set of size $K+1$. In this process, the responsibility of
every element in $S$ is updated accordingly. Next, we find the element with the
largest responsibility in the expanded set and remove that element from the set,
shrinking the set size back to $K$. Again, the responsibility of
every element in $S$ should be updated. Algorithm~\ref{algo:interchange} shows
the pseudo-code for this approach. The theorem below proves the correctness of
the approach.

\begin{theorem}
For  $s_i \in S$, if replacing $s_i$
with a new element $t$ reduces the optimization objective of $S$, applying 
{\tt\small Expand} followed by {\tt\small Shrink} in Algorithm~\ref{algo:interchange}
replaces $s_i$ with $t$. Otherwise, $S$ remains the same.
\end{theorem}

\begin{proof}
Let $\tilde{\kappa}(S)$ indicate $\sum_{s_i,s_j \in S, i < j}
\tilde{\kappa} (s_i,s_j)$.  Also, define $S_- = S - \{s_i\}$ and $S_+ = S +
\{t\}$. We show that if the optimization objective before the replacement,
namely $\tilde{\kappa}(S_- + \{s_i\})$, is larger than the optimization
objective after the replacement, namely $\tilde{\kappa}(S_- + \{t\})$, then the
responsibility of the existing element $s_i$ in an expanded set,
$\rspn_{S_+}(s_i)$, is also larger than the responsibility of the new element
$t$ in the expanded set, $\rspn_{S_+}(t)$. The proof is as follows:
\begin{align*}
&\tilde{\kappa}(S_- + \{s_i\}) > \tilde{\kappa} (S_- + \{t\}) \\
\iff &\sum_{s_j \in S_-} \tilde{\kappa}(s_i, s_j)
> \sum_{s_j \in S_-} \tilde{\kappa}(t, s_j) \\
\iff &\tilde{\kappa}(s_i,t) + \sum_{s_j \in S_-} \tilde{\kappa}(s_i, s_j)
     > \tilde{\kappa}(s_i,t) + \sum_{s_j \in S_-} \tilde{\kappa}(t, s_j) \\
\iff &\rspn_{S_+}(s_i) > \rspn_{S_+}(t).
\end{align*}
Since the responsibility of $s_i$ is larger than that of $t$ in the expanded set
$S_+$, the {\tt\small Shrink} routine will remove $s_i$. If no element 
exists whose responsibility is larger than that of $t$, then $t$ is removed
by this routine and $S$ remains the same.
\end{proof}

In both  the {\tt\small Expand} and {\tt\small Shrink} routines, the responsibility of each
element is updated using a single loop, so both routines take $O(K)$
computations whenever a new data point is considered. Thus, scanning the entire dataset and applying these
two routines will take $O(NK)$ running time.

The \interchange algorithm, if it runs until no replacement decreases the
optimization objective, has the following theoretical bound.

\begin{theorem}
Let's say that the sample obtained by \interchange is $S_{int}$, and the optimal
sample is $S_{opt}$. The quality of $S_{int}$, or the optimization
objective, has the following upper bound:
\begin{align*}
&\frac{1}{K(K-1)} \sum_{s_i, s_j \in S_{int};\; i<j} \tilde{\kappa}(s_i,s_j) \\
&\le \frac{1}{4}
+ \frac{1}{K(K-1)} \sum_{s_i, s_j \in S_{opt};\; i<j} \tilde{\kappa}(s_i,s_j)
\end{align*}
In the expression above, we compare the difference between the averaged
optimization objectives.
\end{theorem}

\begin{proof}
Due to the submodularity of VAS, which we show in our technical report~\cite{vas2015arxiv},
we can apply the result of Nemhauser, \etal~\cite{nemhauser1978analysis} and
obtain the result above.
\end{proof}

Ideally, \interchange should be run until no more valid replacements are possible. 
However,  in
practice, we observed that even running the algorithm for half an hour
  produces a high quality sample. When more time is permitted, the algorithm will
continuously improve the sample quality until convergence.

\ph{Speed-Up using the Locality of Proximity function}
Proximity functions such as $\exp(-\|x-y\|^2/\epsilon^2)$ have a property called
\emph{locality}.
The locality property of a proximity function indicates that its value becomes
negligible when the distance between
the two data points is not close---an idea also used in accelerating other
algorithms~\cite{krause2008near}. For example,
our proximity function value is $1.12 \times 10^{-7}$
when the distance between the two points is $4
\epsilon$; thus, even though we ignore pairs whose distance is larger than a
certain threshold, it will not affect the final outcome much. Exploiting this
property, we can make the {\tt\small Expand} and {\tt\small Shrink} operations much faster
by only considering the data points that are close enough to new data points.
For a proximity check, our implementation used R-tree.

\ignore{
\ph{Implementation Details}
We use the following approach to make the Expand and Shrink operations exploit
proximity function's locality more quickly.
Because checking every data
point in a sample to identify the data points close to a new data point incurs
the same amount of time complexity, we create a grid with its cell width equal
to $\epsilon$. Next, we always keep all the sample data in the grid so
that when a new data point is considered for a valid replacement, 
we can
update just the responsibilities of the data points in the cells near the new data
point. \ignore{\tofix{Note that}, when the sample size $K$ is small, \tofix{we consider more
cells near the new data point} \barzan{what does this sentence mean?} because the distances between the data
points tend to be large.}
To quickly identify the data points with the largest
responsibility in the {\tt Shrink} operation, we maintain a max-heap. With
this approach using the proximity function's locality, we can speed up \interchange more
than 50 times when the sample size is 10,000 or larger.
}

% New section
%\input{maxcover}

 %!TEX root = viz_icde16.tex

\section{Extending VAS: Embedding Density}
\label{sec:density}

\delete{We developed our basic version of \vas (Problem~\ref{prob:vas}), with the goal
of minimizing the visualization quality degradation stemming from the sampling
process. The sample obtained by the basic \vas seems to perform well for the
regression tasks as we will see in Section~\ref{sec:user_perf}. However, we found
that, even the scatter plot drawn using the entire dataset might not be a good
choice for other visualization-driven tasks such as density estimation --- in
dense areas, points are overlapped or cluttered, and eventually, it is hard to
grasp how many points are located in those areas. The situation is not
improved when our basic version of \vas is used for scatter plot visualizations, because
our sample is developed to approximate the shape of the original scatter
plot; thus, it does not choose points according to the density of the original
dataset.}

\vas aims to minimize a visualization-driven quality loss, yielding
scatter/map plots that are highly similar to those generated by visualizing the
entire dataset.
However, we need a different strategy
if the user's intention
 is to estimate the density or find clusters
 from the scatter plot.
This is because humans cannot visually distinguish multiple data points on a scatter plot
if they are duplicates or extremely close to one another. This can make it hard to visually estimate the number or density of such data points.
One way  to address this is to  account for the number of  near-duplicate points in each region. 
For example, points drawn from a dense area can be plotted with a larger legend size or some jitter noise can be used to provide additional density in the plot. 
In other words, the font size of each point or the amount of jitter noise will be proportional to the original density of the region the point is drawn from.
\vas can be easily extended to support such an approach, as follows:

\begin{enumerate}[nolistsep,itemsep=2pt]
\item Obtain a sample using our algorithm for \vas.
\item Attach a counter to every sampled point.
\item While scanning the dataset once more, increase a counter if its associated
sampled point is the nearest neighbor of the data point that was just scanned.
\end{enumerate}
With these extra counters, we can now
visualize the density of areas (each of which is represented by its nearest sampled
point), \eg, using different  dot sizes or by adding jitter noise in
proportion to each point's density.  (See Section~\ref{sec:user} for a scatter
plot example.)

Note that the above process only adds an extra column to the database and, therefore, does
not alter our basic \interchange algorithm for \vas. Also, this extension does
not require any additional information from users.

\delete{On the other hand, we can also show that the sample obtained from \vas is the
\emph{optimal} sample for displaying the density information
when the sample size is constrained to $K$,
and kernel regression~\cite{nadaraya:1964,watson:1964} is used as a way to
measure the quality of estimation for unobserved points. See our technical
report~\cite{vas2015arxiv}
to find that \vas is an optimal sample for regression when the popular
Nadaraya-Watson estimator is used.}

Note that, for the above density embedding process,
a special data structure such as a k-d tree \cite{bentley1975multidimensional} can
be used to make the nearest neighbor tests more efficient. This is done by
using the sample obtained in the first pass to build a k-d tree, then using the
tree to identify the nearest data points in the sample during the second pass.
Since $k$-d trees
perform the nearest neighbor search in $O(\log K)$, the overall time complexity
for the second pass is $O(N \log K)$.

     % Subsection of Maxcover

% New section: experiment
 %!TEX root = viz_icde16.tex

\section{Experiments}
\label{sec:exp}

We run four types of experiments to demonstrate that \vas\ and
\vas\ with density embedding can produce high-quality plots in less time
than competing methods.
\begin{enumerate}[nolistsep,itemsep=2pt]
  \item We study the runtime of existing visualization
        systems that were introduced in Figure~\ref{fig:intro_latency}.
  \item In a user study, we show that users were more successful
        when they used visualizations produced by \vas\ than with other
        methods.  We also show that user success and our loss function
        were negatively correlated (that is, users were
        successful when our loss function is minimized).
  \item We show that \vas\ could obtain a sample of a fixed quality
        level (that is, loss function level) with fewer data points than
        competing methods.  We demonstrate this over a range of
        different datasets and sample quality levels.
  \item We empirically study the \interchange algorithm: we compare its quality and
      runtime to those of the exact method, examine the relationship
      between runtime and sample quality, and investigate the
    impact of our optimization on runtime.
\end{enumerate}

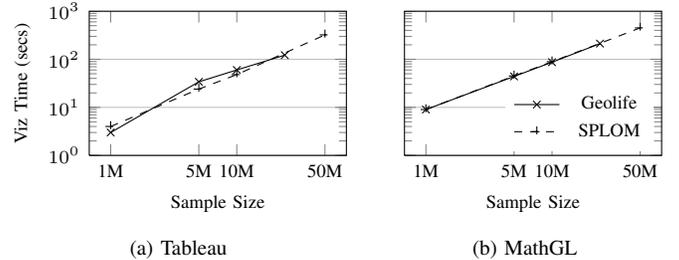
\begin{figure}[!t]
\centering

\begin{minipage}[b]{0.5\textwidth}
\centering

\begin{subfigure}[b]{0.52\textwidth}
\centering
\begin{tikzpicture}
\begin{axis}[
width=50mm,
height=35mm,
xtick={1e6,5e6,1e7,5e7},
xticklabels={1M, 5M, 10M, 50M},
ytick={1,10,100,1000},
ymin=1,
ymax=1000,
xmode=log,
ymode=log,
xlabel=Sample Size,
ylabel=Viz Time (secs),
ymajorgrids,
]
\addplot[mark=x]
table {
x	y
1e6	3
5e6	34
10e6	60
24e6	122
};
\addplot[mark=+,dashed]
table {
x	y
1e6	4
5e6	24
10e6	48
50e6	322
};
%\addlegendentry{Geolife}
%\addlegendentry{SPLOM}
\end{axis}
\end{tikzpicture}
\caption{Tableau}
\end{subfigure}
~
\begin{subfigure}[b]{0.4\textwidth}
\centering
\begin{tikzpicture}
\begin{axis}[
width=50mm,
height=35mm,
xtick={1e6,5e6,1e7,5e7},
xticklabels={1M, 5M, 10M, 50M},
ytick={1,10,100,1000},
yticklabels={,,,},
ymin=1,
ymax=1000,
xmode=log,
ymode=log,
xlabel=Sample Size,
%ylabel=Viz Time (sec),
%legend style={draw=none,at={(1.1,1)},anchor=north west},
legend style={draw=none},
legend pos=south east,
ymajorgrids,
]
\addplot[mark=x]
table {
x	y
1e6	9
5e6	44
10e6	87
24e6	212
};
\addplot[mark=+,dashed]
table {
x	y
1e6	9
5e6	45
10e6	89
50e6	450
};

\addlegendentry{Geolife}
\addlegendentry{SPLOM}
\end{axis}
\end{tikzpicture}
\caption{MathGL}
\end{subfigure}

\end{minipage}

\vspace{-2mm}
\caption{Time to produce plots of various sizes using existing visualization systems.}
\label{fig:viz_latency}
\end{figure}

All of our experiments were performed using two datasets: the Geolife dataset
and the SPLOM dataset.  The
Geolife dataset was collected by
Microsoft
Research~\cite{zheng2008understanding}. It
contained {\sf\small latitude, longitude, elevation} triples from GPS
loggers, recorded mainly around Beijing.  Our full database contained
24.4M tuples.  We also used SPLOM, a synthetic dataset generated from
several Gaussian distributions that had been used in previous visualization
projects~\cite{liu2013immens,kandel2012profiler}.  We used parameters
identical to previous work, and generated a dataset of five columns
and 1B tuples.  We performed our evaluations on an Amazon EC2 memory instance
(r3.4xlarge) which contained 16 cores and 122G memory.

\subsection{Existing Systems are Slow}
\label{sec:latency}

Our discussions in this paper are based on the idea that plotting a scatter
plot using an original dataset takes an unacceptably long period of time.  We tested two state-of-the-art
systems: Tableau~\cite{tableauurl} and MathGL~\cite{mathglurl}. Tableau is one
of the most popular commercial visualization software available on Windows, and
MathGL is an open source scientific plotting library implemented in C++.
We tested both the Geolife and SPLOM datasets.  The results are shown
in Figure~\ref{fig:viz_latency}.

In both
systems, the visualization time includes (1) the time to load data from SSD storage (for
MathGL) or from memory (for Tableau)
and (2) the time to render the data into a plot.  We can
see that even when the datasets contained just 1M records, the
visualization time was more than the 2-second interactive limit.
Moreover, visualization time grew linearly with sample
size.

\ignore{
Table~\ref{tab:viz_online} compares the sample-based visualization system and
the system without any sampling. To process a sample of size 100K, our proposed
method requires about an hour of pre-processing, but once the processing is
finished, the visualization in query times takes only 1.2 seconds.
Without sampling, users of visualization systems should be prepared to wait over
3 minutes whenever requesting a scatter plot, or even longer if the size of the
original dataset is larger.
Notably, the quality of the samples produced by our method were very close to that of the
original dataset, indicating minimal quality degradation as a result of
sampling.
}

\subsection{User Success and Sample Quality}
\label{sec:user}

In this section we make two important experimental claims about user
interaction with visualizations produced by our system.  First, 
users are more successful at our specified goals when using \vas-produced outputs than when using
outputs from uniform random sampling or stratified sampling.  Second,
user success and our loss function --- that is, our measure of
sample quality --- are correlated.  We validate
these claims with a user study performed using Amazon's Mechanical
Turk system.

\subsubsection{User Success}
We tested three user goals: regression, density estimation, and clustering.
\label{sec:user_perf}

\begin{figure}[t]
\centering

\begin{subfigure}[b]{0.48\columnwidth}
\centering
\includegraphics[width=0.9\columnwidth,height=0.7\columnwidth]{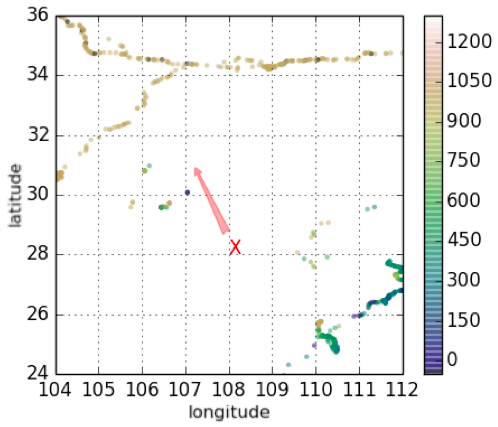}
\caption{Stratified Sampling}
\end{subfigure}
~
\begin{subfigure}[b]{0.48\columnwidth}
\centering
\includegraphics[width=0.9\columnwidth,height=0.7\columnwidth]{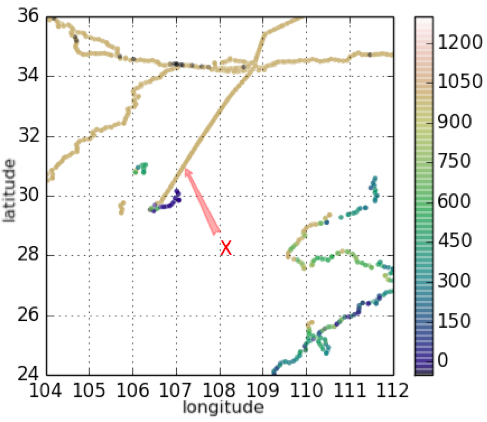}
\caption{\vas}
\end{subfigure}

\vspace{-2mm}

\caption{Example figures used in the user study for the regression task.
We asked the altitude of the location pointed by `X'. The left was
generated by stratified sampling and the right was generated by \vas.}
\label{fig:regression}
\end{figure}

\begin{figure}[t]
\centering

\includegraphics[width=0.9\columnwidth,height=35mm]{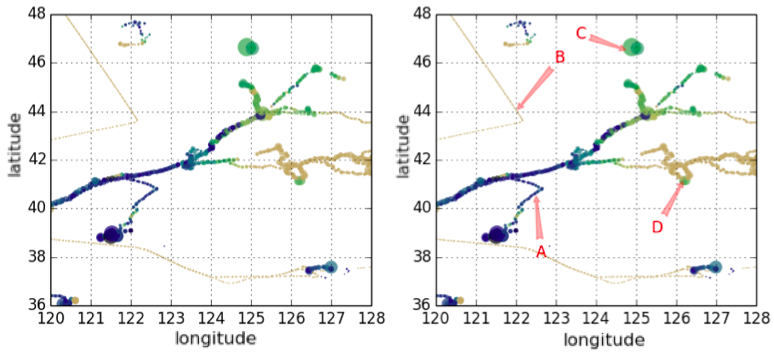}

\caption{An example figure used in the user study for the density estimation task. This
figure was generated using \vas\ with density embedding.
The left-hand image is the original figure.  The right-hand image
contains four test markers, used to ask users to choose the densest
area and the sparsest areas.}
\label{fig:density}
\end{figure}

\ph{Regression} To test user success in the regression task, we gave each user a
sampled visualization from the Geolife data.  We asked the users to estimate the
altitude at a specified latitude and longitude.  Naturally, the
more sampled data points that are displayed near the location in question, the more
accuracy users are likely to achieve.  Figure~\ref{fig:regression}
shows two examples of test visualizations given to users for the
regression task (users were asked to estimate the altitude of the
location marked by `X').  We gave each user a list of four possible
choices: the correct answer, two false answers, and ``I'm not sure''.

We tested \vas,
random uniform sampling, and stratified sampling.  We generated a test
visualization for each sampling method at four distinct sample sizes ranging from
100 to 100K.  For each test visualization, we zoomed into six
randomly-chosen regions and picked a different test location for each region.  Thus, we 
had 72 unique test
questions (3 methods * 4 sample sizes * 6 locations).  We gave each
package of 72 questions to 40 different users 
and averaged the number of correct answers over each distinct question.
To control for worker quality, we filtered out users who failed to correctly
answer a few trivial ``trapdoor'' questions.  

The uniform random sampling method chooses $K$ data points
purely at random, and as a result, tends to choose more data points
from dense areas.  We implemented the single-pass reservoir method for
simple random sampling.  Stratified sampling divides a domain into
non-overlapping bins and performs uniform random sampling for each bin. Here,
the number of the data points to draw for each bin is determined in the most
balanced way. For example, suppose there are two bins and we want a sample of
size 100. If there are enough data points to sample from those two bins, we
sample 50 data points from each bin. Otherwise, if the second bin only has 10
available data points, then we sample 90 data points from the first bin, and 10
data points from the second bin.  Stratified sampling is a straightforward method that avoids uniform
random sampling's major shortcoming (that is, uniform random sampling draws most of its data
points from the densest areas).  In our experiment, stratified
sampling divided the domain of Geolife into 100 exclusive bins and
performed uniform random sampling for each bin using the reservoir
method.

%\begin{table}[t]
%\centering
%\small
%\begin{tabular}{|l|r|r|r|}
%\hline
%Sampling Method & Regression & Density & Clustering \\ \hline
%Simple random sampling & 0.303 & 0.703 & 0.721 \\
%Stratified sampling    & 0.312 & 0.740 & 0.379 \\
%VAS                   & \textbf{0.687} & 0.185 & 0.695 \\
%VAS with density       & N/A   & \textbf{0.930} & \textbf{0.854} \\ \hline
%\end{tabular}
%\caption{User performance comparison in regression, density estimation, and
%clustering tasks.}
%\label{tab:user}
%\end{table}

\begin{table}[t]
%\caption{User performance in the regression, density estimation, and clustering
%tasks.  Uniform is uniform random sampling, and Stratified is stratified sampling.}
\caption{\textsc{User Performance in Three Tasks}}
\label{tab:user}

\centering
\small

\begin{subfigure}[b]{0.5\textwidth}
\centering
\begin{tabular}{|r|r|r|r|}
\hline
Sample size & Uniform & Stratified & VAS   \\ \hline
100   & 0.213   & 0.225      & 0.428 \\
1,000 & 0.260   & 0.285      & 0.637 \\
10,000   & 0.215   & 0.360      & 0.895 \\
100,000  & 0.593   & 0.644      & 0.989 \\ \hline
Average   & 0.319   & 0.378      & \textbf{0.734} \\ \hline
\end{tabular}
\caption{Regression}
\end{subfigure}

\vspace{4mm}

\begin{subfigure}[b]{0.5\textwidth}
\centering
\begin{tabular}{|r|r|r|r|r|}
\hline
Sample size & Uniform & Stratified & VAS   & VAS w/ density \\ \hline
100   & 0.092   & 0.524      & 0.323 & 0.369 \\
1,000 & 0.628   & 0.681      & 0.311 & 0.859 \\
10,000   & 0.668   & 0.715      & 0.499 & 0.859 \\
100,000  & 0.734   & 0.627      & 0.455 & 0.869 \\ \hline
Average   & 0.531   & 0.637      & 0.395 & \textbf{0.735} \\ \hline
\end{tabular}
\caption{Density Estimation}
\end{subfigure}

\vspace{4mm}

\begin{subfigure}[b]{0.5\textwidth}
\centering
\begin{tabular}{|r|r|r|r|r|}
\hline
Sample size & Uniform & Stratified & VAS   & VAS w/ density \\ \hline
100   & 0.623   & 0.486      & 0.521 & 0.727 \\
1,000 & 0.842   & 0.412      & 0.658 & 0.899 \\
10,000   & 0.931   & 0.543      & 0.845 & 0.950 \\
100,000  & 0.897   & 0.793      & 0.864 & 0.965 \\ \hline
Average   & 0.821   & 0.561      & 0.722 & \textbf{0.887} \\ \hline
\end{tabular}
\caption{Clustering}
\end{subfigure}

\end{table}

Table~\ref{tab:user}(a) summarizes user success in the regression task. The result shows that
users achieved the highest accuracy in the regression task when
they used \vas, significantly outperforming other sampling methods.

\ph{Density Estimation} For the density estimation
task, we created samples whose sizes ranged 100-100K using four different
sampling methods: uniform random sampling, stratified sampling, VAS,
and VAS with density embedding.  Using those samples, we chose 5 
different zoomed-in areas.  For each zoomed-in area, we asked users to identify the densest
and the sparsest areas among 4 different marked locations.
Figure~\ref{fig:density} shows an example visualization shown to a
test user.  As a result, we generated 80 unique visualizations.  We
again posed the package of 80 questions to 40 unique users, and again
filtered out users who failed to answer easy trapdoor questions.

The result of the density estimation task is shown in Table~\ref{tab:user}(b).
Interestingly, the basic \vas\ method \emph{without} density estimation
yielded very poor results.
%As noted earlier in Section~\ref{sec:density},
%\vas\ ignores the underlying density of points
%in the original dataset.
However, when we augmented the sample with density
embedding, users obtained even better success than with uniform random sampling. One
of the reasons that `VAS with density' was superior to uniform random sampling
was because we not only asked the users to estimate the densest area, but also
asked them to estimate the sparsest area of those figures. The figures generated by
uniform random sampling typically had few points in sparse areas, making it
difficult to identify the sparsest area.

\ph{Clustering} Lastly, we compared user performance in
the clustering task. Since the Geolife dataset did not have ground-truth
for clustering, we used synthetic datasets that we generated using Gaussian
distributions instead. Using two-dimensional Gaussian distributions with different covariances, we
generated 4 datasets, 2 of which were generated from 2 Gaussian
distributions and the other 2 were generated from a single Gaussian
distribution.  (This dataset was similar to SPLOM, which unfortunately
has a single Gaussian cluster, making it unsuitable for this experiment.)

Using the same 4 sampling methods that were used in the density estimation task,
we created samples whose sizes ranged 100-100K, and tested if users could correctly identify the
number of underlying clusters given the figures generated from those samples.
In total, we created 64 questions (4 methods, 4 datasets, and 4 sample
sizes).  We again asked 40 Mechanical Turk users (or simply Turkers) and filtered out bad workers.

Table~\ref{tab:user}(c) summarizes the result of the clustering task.
As in the density estimation task, `\vas\ with density' allowed users to
be more successful than they were with visualizations from uniform
random sampling.  Although \vas\ without density did not 
perform as well as uniform random sampling, it produced a roughly comparable
score.

We think the reason VAS \emph{without} density estimation showed comparable
performance was that we used no more than 2 Gaussian distributions for data
generation, and the Turkers could recognize the number of the underlying
clusters from the outline of the sampled data points. For example, if the data
were generated from two Gaussian distributions, the data points sampled by VAS
would look like two partially overlapping circles. The Turkers would have shown poorer
performance if there was a cluster surrounded by other clusters.

On the other hand, stratified sampling did poorly in this clustering task
because it performed a separate random sampling for each bin, i.e.,
the data points within each bin tend to group
together, and as a result, the Turkers found that there were 
more clusters than actually existed.

\subsubsection{Correlation with Sample Quality}
\label{sec:corr}

In this section, we test whether the \vas's optimization criterion of
$\loss(S)$ had a close relationship to our visualization users' success
in reaching their end goals.  If they were highly correlated, we have some empirical evidence
that samples which minimize the \vas\ loss function will yield
useful plots.

In particular, we examined this relationship for the case of
regression.  For each combination of sample size and sampling method,
we produced a sample and corresponding visualization.  We computed $\loss(S)$
using the expression in Equation~\ref{eq:loss}.  We then
measured the correlation between the \emph{loss} and average user
performance on the regression task for that visualization.

To compute the loss (Equation~\ref{eq:loss}), which includes integration, we used
the Monte Carlo technique using 1,000 randomly generated points in the domain of
the Geolife dataset. For this process, we determined that randomly generated points
were within the domain if there existed any data point in the original dataset whose
distance to the randomly generated data points was no larger than 0.1. Now, the
integral expression was replaced with a summation as follows:
\[
\loss(S) =
\frac{1}{1000} \sum_{i = 1}^{1000} \frac{1}{\sum_{s_i \in S} \kappa(x_i, s_i)}.
\]
This \emph{loss} computed above is the \emph{mean} of one
thousand values.  One problem we encountered in computing the mean was
that the point-loss often became so large that {\tt\small double} precision could not
hold those quantities. To address this, we used the \emph{median} instead of
the \emph{mean} in this section because the median is less sensitive to outliers.
Note that the median is still valid for a correlation analysis because we did
not observe any case where a sample with a larger mean has a smaller median
compared to another sample.

Next, to compare loss in a more consistent
way, we computed the following quantity:
\[
\text{log-loss-ratio}(S) = \log_{10} \left[ \frac{\loss(S)}{\loss(D)} \right]
\]
where $D$ is the original dataset. $\loss(D)$ is the lowest \emph{loss} that a
sample can achieve; thus, samples with log-loss-ratios close to zero can be
regarded as good samples based on this metric.

Next we examined the relationship between a sample's
log-loss-ratio and the percentage of the questions that were
correctly answered in the regression task using the sample's corresponding
visualization.  If the two metrics yield similar rankings of the
sampled sets, then
the \vas\ optimization criterion is a good guide for producing
end-user visualizations.  If the two metrics yield uncorrelated
rankings, then our \vas\ criterion is faulty.

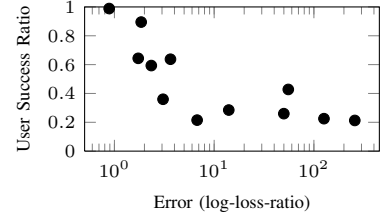
\begin{figure}[t!]
\centering
\begin{tikzpicture}
\begin{axis}[
width=55mm,
height=35mm,
xmode=log,
ymin=0,
ymax=1,
xlabel=Error (log-loss-ratio),
ylabel=User Success Ratio,
]
\addplot[only marks,mark size=2.0,draw=black,fill=black]
table[x=x,y=y] {
x	y
%2.41001289	0.2127659574
%1.69849058	0.2598870057
%0.8292588549	0.2150537634
%0.3698119953	0.5934065934
%2.100970947	0.2251308901
%1.146512813	0.2849462366
%0.4874937887	0.3597883598
%0.238785856	0.6436170213
%1.743204513	0.4278074866
%0.5620114794	0.6373626374
%0.2684418896	0.8950276243
%-0.05253851175	0.9888888889
257.0472076	0.2127659574
49.94483465	0.2598870057
6.749301902	0.2150537634
2.343214226	0.5934065934
126.1743124	0.2251308901
14.01240926	0.2849462366
3.072513422	0.3597883598
1.732949296	0.6436170213
55.36107478	0.4278074866
3.647635883	0.6373626374
1.85541853	0.8950276243
0.8860566477	0.9888888889
};

%\addlegendentry{uniform}
%\addlegendentry{stratified}
%\addlegendentry{VAS}

\end{axis}
\end{tikzpicture}
\caption{The relationship between the loss and user performance on the regression
task. The samples with smaller losses resulted in better success
ratios in general in the regression task.}
\label{fig:var_user}
\end{figure}

Figure~\ref{fig:var_user} depicts the results. The figure clearly shows
the negative correlation between the loss
and user success ratio in the regression task. Because the X-axis of the figure
is the loss
function that we aim to minimize to obtain a good sample, the negative correlation
between the two metrics shows the validity of our problem formulation.

Also, when we computed Spearman's rank correlation
coefficient\footnote{Spearman's rank correlation coefficient produces $-1.0$
for pairs of variables that are completely negatively correlated, and $1.0$ for
pairs of variables that
are completely positively correlated.}, the correlation coefficient was $-0.85$, indicating a strong
negative correlation between user success and the log-loss-ratio.  (Its p-value
was $5.2 \times 10^{-4}$.) Put another way, minimizing our loss function for a
sample should do a good job of maximizing user success on the resulting
visualizations.  This result indicates that the problem formulation in
Section~\ref{sec:formulation} and the intuition behind it was largely valid.

%Figure~\ref{fig:var_user} shows the relationship between sample variance and 
%user performance in the regression task. In this figure,
%we used a log-log function on variance-ratio, because
%the quantity often became very large when the sample sizes were small. In
%general, we can observe that there is a strong negative correlation between the
%two metrics, \ie, the smaller the sample variance is, the better the user performance
%is in the regression task. \tofix{The Spearman's rank correlation coefficient was
%$-0.85$ with p-value of $5.2 \times 10^{-4}$.}

\subsection{VAS Uses a Smaller Sample}
\label{sec:var_comp}

This section shows that \vas\ can produce a better sample than
random uniform sampling or stratified sampling.  That is, for a fixed
amount of visualization production time, its quality (loss function value)
is lower; or, that for a fixed quality level (loss function value),
it needs less time to produce the visualization.  (The visualization
production time is linear with the number of data points.)

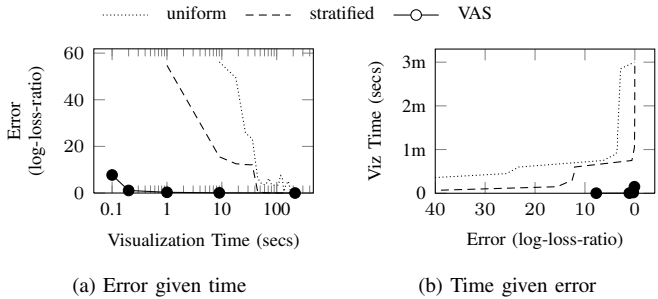
\begin{figure}[t]
\centering

\begin{subfigure}[b]{0.48\columnwidth}
\centering
\begin{tikzpicture}
\begin{axis}[
width=45mm,
height=35mm,
xtick={0.1,1,10,100},
xticklabels={0.1,1,10,100},
ymin=0,
xmode=log,
xlabel=Visualization Time (secs),
ylabel style={align=center},
ylabel=Error\\(log-loss-ratio),
legend style={at={(0.0,1.1)},anchor=south west,draw=none},
legend columns=3,
]
\addplot[mark=none,densely dotted]
table[x=x,y=y] {
x	y
9	56.21720384
18	49.43553602
27	25.72692653
36	23.24105346
45	6.322439246
54	3.745373835
63	3.48268737
72	6.319681812
81	3.468174433
90	3.369245473
99	3.331254949
108	3.409924135
117	7.543737908
126	6.921320444
135	0.606655743
144	2.831318272
%153	53.86807514
162	5.145872202
171	2.850050387
180	0.08898386043
216	0
};
\addplot[mark=none,densely dashed]
table {
x	y
1	54.7470314
9	15.55029342
18	12.56626193
27	12.19196802
36	12.08425294
45	0.5979984235
54	0.3251792837
63	0.07802466885
72	0.06130005143
81	0.04546863489
90	0.03759601179
99	0.03065569413
108	0.02469522715
117	0.0203083769
126	0.01643319329
135	0.01308386191
144	0.01083653746
153	0.01027288458
162	0.007443594562
171	0.006306697363
180	0.004595751689
216	0
};
\addplot[mark=*]
table {
x	y
0.100	7.738679375
0.200	1.098754154
1	0.3141138576
9	0.07802466885
216	0
};

\addlegendentry{uniform}
\addlegendentry{stratified}
\addlegendentry{VAS}

\end{axis}
\end{tikzpicture}
\caption{Error given time}
\end{subfigure}
~
\begin{subfigure}[b]{0.48\columnwidth}
\centering
\begin{tikzpicture}
\begin{axis}[
width=45mm,
height=35mm,
ytick={0,60,120,180,240},
yticklabels={0,1m,2m,3m,4m},
ymin=0,
%xmode=log,
xmax=40,
x dir=reverse,
xlabel=Error (log-loss-ratio),
ylabel=Viz Time (secs),
legend style={at={(1.1,1)},anchor=north west,draw=none}
]
\addplot[mark=none,densely dotted]
table[x=x,y=y] {
x	y
56.21720384	9
49.43553602	18
25.72692653	27
23.24105346	36
6.322439246	45
3.745373835	54
3.48268737 	63
%6.319681812	72
3.468174433	81
3.369245473	90
3.331254949	99
2.850050387	171
0.088983860	180
%0          	216

};
\addplot[mark=none,densely dashed]
table [x=x,y=y] {
x	y
54.747031	1
15.550293	9
12.566261	18
12.191968	27
12.084252	36
0.5979984	45
0.3251792	54
0.0780246	63
0.0613000	72
0.0454686	81
0.0375960	90
0.0306556	99
0.0246952	108
0.0203083	117
0.0164331	126
0.0130838	135
0.0108365	144
0.0102728	153
0.0074435	162
0.0063066	171
0.0045957	180
%0        	216

};
\addplot[mark=*]
table [x=x,y=y] {
x	y
7.73867	0.10
1.09875	0.20
0.31411	1
0.07802	9
%0      	216
};

\end{axis}
\end{tikzpicture}
\caption{Time given error}
\end{subfigure}

\caption{Relationship between visualization production time and error
for the three sampling methods.}
\label{fig:variance}
\end{figure}

We used the Geolife dataset and produced samples of various sizes (and
thus, different visualization production times).
Figure~\ref{fig:variance}(a) shows the results when we varied the
visualization time: \vas\ always produced a sample with lower loss
function values (\ie, higher quality) than other methods.  The quality gap
between the methods did not become smaller until after an entire minute
of runtime.  We show a similar result with the other dataset in our technical
report~\cite{vas2015arxiv}

Figure~\ref{fig:variance}(b) shows the same data using a different perspective. We
fixed the loss function value (quality) and measured how long it takes
to visualize the corresponding samples.   Because
the samples generated by our method had much smaller losses compared
to other methods, all of the data points in the figure are in the
bottom right corner.  Competing methods required much more time than
\vas\ to obtain the same quality (loss function value).

\delete{
\mike{The below discussion is totally unconvincing. We really should
have an experiment that directly examines the issue.}
It is possible that one could improve the performance of stratified sampling by
carefully adjusting the number of its bins (which was fixed to $100$ in our
experiments) based on the sample size and the underlying data distribution.
However, even when we \addnew{increased the strata count for stratified sampling to
the one equal to the sample size, stratified sampling still
showed poor performance especially for zoomed-in plots (as depicted in
Figure~\ref{fig:intro}).}}

\delete{tested stratified sampling with various bin numbers, it always
resulted in higher loss values (thus, poorer quality) compared to \vas. Also,
\vas does not require any repeated testing with different parameters, which
indicates that data analysts need less \emph{a priori} knowledge when our
sampling technique is used.}

\subsection{Algorithmic Details}
\label{sec:runtime_anal}
\label{sec:exact}
\label{sec:runtime_comp}

We now examine three internal qualities of the \vas\ technique:
approximate vs.~exact solution, runtime analysis, and optimization
contributions.

\ph{Exact vs.~Approximate}
The NP-hardness of \vas supports the need for an approximation algorithm. This
section empirically examines the NP-hardness of VAS.

We think one of the best options for obtaining an exact solution to VAS is by
converting the problem to an instance of integer programming and
solving it using a standard library. Refer to our report~\cite{vas2015arxiv} for converting VAS to
an instance of Mixed Integer Programming (MIP). We used the GNU Linear
Programming Kit~\cite{makhorin2004glpk} to solve the converted MIP problem.

\begin{table}[t]
%\caption{Loss and runtime comparison of the algorithms for \vas. Uniform random
%sampling (Random) is included for reference. MIP stands for Mixed Integer
%Programming, a procedure for obtaining exact solutions to VAS.
%$K$ is fixed to 10 in all the experiments above.}
  \caption{\textsc{Loss and Runtime Comparison}}
\label{tab:mip}

\centering
\small
\begin{tabular}{|c|l|r|r|r|}
\hline
N                   & Metric          & MIP     & Approx. VAS  & Random   \\ \hline
\multirow{3}{*}{50} & Runtime         & 1m 7s   & 0s       & 0s       \\
                    & Opt. objective  & 0.160   & 0.179    & 3.72     \\
                    & $Loss(S)$ & 1.5e+26 & 1.5e+26  & 2.5e+29  \\ \hline
\multirow{3}{*}{60} & Runtime         & 1m 33s  & 0s       & 0s       \\
                    & Opt. objective  & 0.036   & 0.076    & 3.31     \\
                    & $Loss(S)$ & 3.8e+11 & 1.6e+16  & 2.5e+29  \\ \hline
\multirow{3}{*}{70} & Runtime         & 14m 26s & 0s       & 0s       \\
                    & Opt. objective  & 0.047   & 0.048    & 3.02     \\
                    & $Loss(S)$ & 1.8e+13 & 1.8e+13  & 9.45e+33 \\ \hline
\multirow{3}{*}{80} & Runtime         & 48m 55s & 0s       & 0s       \\
                    & Opt. objective  & 0.043   & 0.048    & 2.25     \\
                    & $Loss(S)$ & 8.5e+13 & 1.8e+13 & 9.4e+35  \\ \hline
\end{tabular}
\end{table}

Table~\ref{tab:mip} shows the time it took to obtain exact solutions to VAS
with datasets of very small sizes. The sample size $K$ was fixed to 10 in
all of the experiments in the table.  According to the result, the
exact solutions to VAS showed better quality, but the procedure to
obtain them took considerably longer. As shown, obtaining an exact
solution when $N = 80$ took more than 40 minutes, whereas the time it
took by other sampling methods was negligible.  Clearly, the exact
solution is not feasible except for extremely small data sizes.

\input{tab_runtime}

\ph{Runtime Analysis}
\vas gradually improves its sample quality as more data is read and
processed. As observed in
many statistical optimization routines, \vas offers good-quality plots long
before reaching its optimal state.
To investigate this phenomenon, we measured the relationship between ``processing time'' and
``visualization quality.'' The result is shown in Figure~\ref{fig:runtime}. Note that the
Y-axis of the figure is the objective\footnote{We scaled the objectives
appropriately for a clearer comparison.} of our minimization problem; thus, the
lower the objective, the higher the associated visualization's quality. 
In this experiment, we used the Geolife dataset. Figure~\ref{fig:runtime}  demonstrates that our
\interchange algorithm improved the visualization quality quickly at its initial
stages, and the improvement rate slowed down gradually. Notably, \vas produced
low-error plots within only tens of minutes of processing time. The storage
overhead of our algorithm is only $O(K)$, where $K$ is the sample size.

\delete{
Note that the \vas method is designed for offline uses, and we want to see 
if our algorithm can finish sampling processes in a reasonable amount of
time even for very large datasets. Figure~\ref{fig:runtime}
reports our runtime experiments with the Geolife dataset and the SPLOM dataset.
To obtain the samples of sizes 100K and 1M for the Geolife dataset, our
algorithm took about 1 and 7 hours respectively. For the SPLOM dataset of 1B
tuples, our algorithm took about 7 and 12 hours respectively. Although we showed
that our current implementation can process large datasets, we are working to
further improve its runtime by performing more aggressive approximations
including pre-shuffling and hierarchical approaches.\footnote{We observed that
the runtime can be improved greatly by running multiple \vas in parallel after
discretizing the space into multiple bins; those multiple samples are combined
and another \vas is run on the merged set.}}

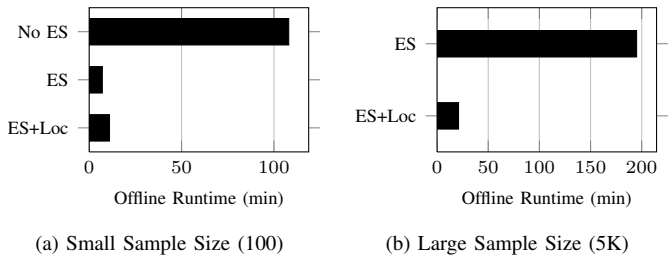
\begin{figure}[t]
  \centering

  \pgfplotsset{mybar/.style={
      width=45mm,
      height=35mm,
      xbar,
      xmin=0,
      ymin=0.5,
      ymax=3.5,
      ytick={1,2,3},
      yticklabels={ES+Loc,ES,No ES},
      xlabel=Offline Runtime (min),
      xmajorgrids,
  }}

  \begin{subfigure}[b]{0.48\columnwidth}
    \begin{tikzpicture}
      \begin{axis}[mybar]
        \addplot[fill=black] coordinates {
          (11,1)
          (7.12,2)
          (108,3)
        };
      \end{axis}
    \end{tikzpicture}
    \caption{Small Sample Size (100)}
  \end{subfigure}
  ~
  \begin{subfigure}[b]{0.48\columnwidth}
    \begin{tikzpicture}
      \begin{axis}[mybar,
          ytick={1,2},
          yticklabels={ES+Loc,ES},
          ymax=2.5,
        ]
        \addplot[fill=black] coordinates {
          (21,1)
          (195,2)
        };
      \end{axis}
    \end{tikzpicture}
    \caption{Large Sample Size (5K)}
  \end{subfigure}

  \caption{Runtime comparison of different levels of optimizations. For
  this experiment, we used the Geolife dataset. ES+Loc indicates that both
Expand/Shrink (ES) operation and the locality of a proximity function were used.}
  \label{fig:breakdown}
\end{figure}

\ph{Optimization Contribution}
\addnew{To quantify the impact of our optimization efforts on the
runtime reduction, we measured the runtime of three different settings:
\begin{enumerate}[nolistsep,itemsep=2pt]
  \item No Expand/Shrink (No ES): This is the most basic configuration that does
    not use the Expand/Shrink approach, but instead compares the responsibility
    when a new point is switched with another one in the sample.
  \item Expand/Shrink (ES): This version uses the Expand/Shrink operation,
    reducing the time complexity by $O(K)$, where $K$ is the sample size.
  \item Expand/Shrink+Locality (ES+Loc): This version uses an additional R-tree
    to speed up the Expand/Shrink operations. This version is possible due to
    the locality of our loss function.
\end{enumerate}
Figure~\ref{fig:breakdown} shows the results. When the sample size was relatively small (100), the second
approach (Expand/Shrink), which does not exploit the locality, showed the shortest
runtime due to no extra overhead coming from maintaining an extra R-tree data
structure. However, when the sample size was relatively large (5K), the last
approach (ES+Loc) that exploits the locality of the loss function showed the
fastest runtime. When the user is interested in large samples (more than 10K at
least), the last approach that uses R-tree to exploit locality will be the
most preferable choice. The runtime sensitivity to sample size suggests that in
the future, it may be useful to employ an optimizer that chooses the most
appropriate algorithm setting, given a requested sample size.}

%Our method took approximately 12 hours to
%produce a
%sample of size 1M from the SPLOM dataset of 1B records. The speed of
%\interchange is certainly slower than uniform random sampling which took about
%40 minutes for generating a sample of size 1M from the SPLOM dataset, but
%\interchange could handle a dataset of a billion records in a reasonable amount
%of time.

%The \interchange method took approximately 12 hours to produce a sample of 1M
%records from the SPLOM database of 1B records. While this is slower than the 40
%minutes it took for uniform random sampling to complete the same task,
%\interchange could handle a dataset of a billion records in a reasonable amount
%of time.

\ignore{
\subsection{Time Complexity Analysis}
\label{sec:runtime}

\begin{figure}[t]
\centering

\begin{subfigure}[b]{0.5\textwidth}
\centering
\begin{tikzpicture}
\begin{axis}[
width=60mm,
xmin=0.5,
xmax=3.5,
ymin=0,
ymax=6,
xtick={1,2,3},
xticklabels={100,1000,10000},
ytick={0,2,...,6},
xlabel=Sample Size,
ylabel=$\log_{10} \text{sec}$,
legend style={at={(1.1,1)},anchor=north west,draw=none},
]
\addplot[mark=x]
table {
x	y
1	0.30
2	0.30
3	0.30
};
\addplot[mark=triangle,dashed]
table {
x	y
1	3.13
2	5.08
};
\addplot[mark=*]
table {
x	y
1	1.47
2	2.46
3	3.56
};
\addlegendentry{Random}
\addlegendentry{\greedy}
\addlegendentry{\interchange}
\end{axis}
\end{tikzpicture}
\caption{Runtime w.r.t.~sample size}
\end{subfigure}

\vspace{4mm}

\begin{subfigure}[b]{0.5\textwidth}
\centering
\begin{tikzpicture}
\begin{axis}[
width=60mm,
xtick={1,2,3},
xticklabels={1000,10000,100000},
ytick={-2,0,...,4},
xmin=0.5,
xmax=3.5,
ymin=-2,
ymax=5,
xlabel=Dataset Size,
ylabel=$\log_{10} \text{sec}$,
legend style={at={(1.1,1)},anchor=north west,draw=none},
]
\addplot[mark=x]
table {
x	y
1	-1.70
2	-0.70
3	0.30
};
\addplot[mark=triangle,dashed]
table {
x	y
1	3.07
2	4.1
};
\addplot[mark=*]
table {
x	y
1	0.45
2	1.45
3	2.45
};
\addlegendentry{Random}
\addlegendentry{\greedy}
\addlegendentry{\interchange}
\end{axis}
\end{tikzpicture}
\caption{Runtime w.r.t.~dataset size}
\end{subfigure}

\caption{Time complexity analysis of methods for VAS. We included simple random
sampling (Random) for reference.}
\label{fig:runtime}
\end{figure}

This section shows by experiments the correctness of the time complexity analysis of
the methods for \vas that were presented in Section~\ref{sec:balanced}.
According to our analysis, the time complexity of \greedy is $O(NK^2)$, and the
time complexity of \interchange is $O(NK)$. For the experiments in this section,
we used synthetic datasets generated from two Gaussian distributions to have
tight control over the dataset sizes, but since the number of computations
required to produce a certain size of sample does not depend on the data
distribution in datasets, our empirical analysis in this section should be
applicable to all other datasets.

Since the time complexities of the two algorithms are the functions of dataset
size and sample size, we performed two sets of experiments: first, we fixed the
dataset size to 1 million and changed the sample size from 100 to 10,000;
second, we fixed the sample size to 1,000 and changed the dataset size from
1.000 to 100,000. Figure~\ref{fig:runtime}a and
Figure~\ref{fig:runtime}b show the results of the two sets of experiments
respectively. Note that Y-axis, which shows elapsed time, is in log-scale. In
those experiments, we included simple random sampling for reference, because it
would probably be the fastest sampling method, although the loss of the
samples produced by the method is significantly higher than those by our proposed
methods.

In Figure~\ref{fig:runtime}a, we can observe that as the sample size increase by
a factor of 10, the runtime of \interchange increases by about 1 in log-scale,
and the runtime of \greedy increases by about 2 in log-scale. This tells us that
the time complexity of \interchange is in linear relationship with respect to
the sample size, and the time complexity of \greedy is in quadratic relationship
with respect to the sample size, as analyzed earlier. In
figure~\ref{fig:runtime}b, we can observe that the runtimes of the both methods,
\interchange and \greedy, increase by about 1 in log-scale as the dataset size
increases by a factor of 10. This shows that the time complexities of both
algorithms are the linear functions of dataset sizes as expected.
}

\section{Related Work}
\label{sec:relatedwork}

Support for interactive visualization of large datasets is a fast-growing
area of research interest~\cite{battle:2013, cottam:2013,
heer:2012, wickham:2013, liu2013immens, lins:2013, barnett:2013,
fisher:2012, fisher:2012b, cottam:2010, piringer:2009}, \addnew{along with other
approximate techniques for interactive processing of non-conventional
queries~\cite{mozafari_pvldb2015_ksh}.}  Most of the
work to date has originated from the user interaction community, but
researchers in database management have begun to study the problem.
Known approaches fall into a few different categories.

The most directly related work is that of Battle,
\etal~\cite{battle:2013}.  They proposed ScalaR, a system for {\em
dynamic reduction} of query results that are too large to be
effectively rendered on-screen.  The system examines queries sent from
the visualization system to the RDBMS and if necessary, inserts
aggregation, sampling, and filtering query operators.  ScalaR uses
simple random sampling, and so could likely be improved by adopting
our sampling technique. For bar graphs, Kim \etal~\cite{kim2015rapid}
proposed an order-preserving sampling method, which examines fewer tuples than
simple random sampling.

%  (vision papers and demos: heer:2012,  barnett:2013,)

%  binned aggregation: liu2013immens, lins:2013,
{\em Binned aggregation} approaches~\cite{lins:2013, liu2013immens, wickham:2013}
reduce data by dividing a data domain into tiles or bins, which correspond to
materialized views.  At visualization time, these bins can be selected
and aggregated to produce the desired visualization.  Unfortunately,
the exact bins are chosen ahead of time, and certain operations ---
such as zooming --- entail either choosing a very small bin size (and
thus worse performance) or living with low-resolution results.
Because binned aggregation needs to pre-aggregate all the quantities in
advance, the approach is less flexible when the data is changing, such
as measured temperatures over time; our method does not have such a
problem.

%{\bf One of
%the disadvantages is that the required storage size grows linearly as the number
%of the quantities for pre-computation increases. The drawback becomes more critical
%when the data are expected to change over time, for example, when visualizing
%measured temperatures.}  \mike{Young, I don't understand the previous
%sentence.  Can you clarify?}

Wickham~\cite{wickham:2013} proposed to improve visualization times with a mixture of
binning and summarizing (very similar to binned aggregation) followed
by a statistical smoothing step.  The smoothing step allows the
system to avoid problems of high variability, which arise when the
bins are small or when they contain eccentric values.  However, the
resulting smoothed data may make the results unsuitable for certain
applications, such as an outlier finding. This smoothing step
itself is orthogonal to our work, i.e., when there appears to be high variability in
the sample created by our proposed method, the same smoothing technique can be
applied to present more interpretable results. The smoothing process also
benefits from our method because \vas\ creates a sample much smaller than the
original database, thus, makes smoothing faster.
%{\bf REMIND -- Yongjoo, please
%read over this paper and try to have succinct technical diff with our
%work.}
The {\em abstract rendering pipeline}~\cite{cottam:2013} also maps bins to
regions of data, but the primary goal of this system is to modify the
visualization, not performance.

%  parallel processing: piringer:2009, cottam:2010
{\em Parallel rendering} exploits parallelism in hardware to speed up visual
drawing of the visualization~\cite{cottam:2010, piringer:2009}.  It is
helpful but largely orthogonal to our contributions. {\sc SeeDB} is a system
that discovers the most interesting bar graphs~\cite{vartak:2015} from datasets.

%  incremental viz production: fisher:2012, fisher:2012b
{\em Incremental visualization} proposes a streaming data processing model,
which quickly yields an initial low-resolution version of the user's
desired bitmap~\cite{fisher:2012, fisher:2012b}.  The system continues to process data after showing
the initial image and progressively refines the visualization.  When
viewed in terms of our framework in Section~\ref{sec:vizsystems}, this
method amounts to increasing the sample budget over time and
using the new samples to improve the user's current visualization.
Thus, incremental visualization and sample-driven methods should
benefit from each other.

% What about Scattr (lee:2013)?  What about Mitar from Berkeley?

% 1.  Classify wickham:2013
% 2.  Improve description of abstract rendering
% 3.  Improve description of binned aggregation

%!TEX root = viz_icde16.tex

\section{Conclusions and Future Work}
\label{sec:con}
We have described the \vas\ method for visualization data reduction.
\vas\ is able to choose a subset of the original database that is very
small (and thus, fast) while still yielding a high-quality scatter or
map plot.  Our user study showed that for three common user goals ---
regression, density estimation, and clustering --- \vas\ outputs are
substantially more useful than other sampling methods' outputs with
the same number of tuples.

\delete{We further showed that the \vas\ offline
phase, while more time-consuming than competitors', is still feasible
for large databases.}

\delete{We envision two areas for extending this work in the future.  First,
we want to research some of the connections between our proposed
algorithm and other domains not directly related to visualization,
such as optimal sensor placement~\cite{krause2008near}.  Second,}

We believe our core topic --- data system support for visualization
 tools --- is still in its infancy and entails a range of interesting
    challenges.  In particular, we plan 
 to investigate techniques for rapidly generating visualizations
for other user goals (including outlier detection, trend
identification) and other data types (such as
large networks).  

\section*{Acknowledgment}
This work is in part sponsored by NSF awards ACI-1531752, CNS-1544844,
and III-1553169. The authors are also grateful to Suchee Shah for her great comments on this manuscript.


%\section*{Acknowledgment}



\bibliographystyle{IEEEtran}
%\bibliography{IEEEabrv,viz_icde16}
\bibliography{viz_icde16,mozafari}
%\bibliography{diversity_sampling}  % sigproc.bib is the name of the Bibliography in this case
%\begin{thebibliography}{1}
%\end{thebibliography}

%\appendix
%
%
%% New section: proof on the NP-hardness of VAS
%\input{np_hard}
%
%
%% New section: proof on the submodularity of VAS
%\input{submodular}
%
%
%\input{mip}
%
%
%\section{More Experiments}
%
%\input{var_user}
%
%\input{splom_var}
%
%\input{osm}

\end{document}